\documentclass[review,onefignum,onetabnum]{siamart171218}

\usepackage{amsmath}
\usepackage{amsbsy}
\usepackage{amstext}
\usepackage{amssymb}
\usepackage{babel}
\usepackage{xcolor}
\usepackage{siunitx}
\usepackage{float}
\usepackage{graphicx}
\usepackage{esint}
\usepackage{enumitem}

\newcommand{\chap}{Chapman and Wilmott \cite{chapman2021}~}

\usepackage{lipsum}
\usepackage{amsfonts}
\usepackage{graphicx}
\usepackage{epstopdf}
\usepackage{algorithmic}
\ifpdf
  \DeclareGraphicsExtensions{.eps,.pdf,.png,.jpg}
\else
  \DeclareGraphicsExtensions{.eps}
\fi

\newsiamremark{remark}{Remark}
\newsiamremark{hypothesis}{Hypothesis}
\crefname{hypothesis}{Hypothesis}{Hypotheses}
\newsiamthm{claim}{Claim}

\headers{Homogenisation of nonlinear blood flow in periodic networks}{Y. Ben-Ami, B. D. Wood, J. M. Pitt-Francis, P. K. Maini, and H. M. Byrne}

\title{Homogenisation of nonlinear blood flow in periodic networks: the limit of small haematocrit heterogeneity \thanks{Submitted to the editors \today.
\funding{This work was supported by the Engineering and Physical Sciences Research Council [grant number EP/X023869/1].}}}

\author{
Yaron Ben-Ami\thanks{Wolfson Centre for Mathematical Biology, Mathematical Institute, University of Oxford, Oxford, OX2 6GG, United Kingdom
  (\email{yaron.ben-ami@maths.ox.ac.uk}, \email{philip.maini@maths.ox.ac.uk}, \email{helen.byrne@maths.ox.ac.uk}).}
\and Brian D. Wood\thanks{School of Chemical, Biological, and Environmental Engineering, Oregon State University, Corvallis, OR 97331, United States
  (\email{brian.wood@oregonstate.edu}).}
\and Joe M. Pitt-Francis \thanks{Department of Computer Science, University of Oxford, Oxford, OX1 3QD, United Kingdom 
  (\email{Joe.Pitt-Francis@cs.ox.ac.uk}).}
\and Philip K. Maini\footnotemark[2]
\and Helen M. Byrne\footnotemark[2] \thanks{Ludwig Institute for Cancer Research, University of Oxford, Oxford, OX3 7DQ, United Kingdom}
}

\usepackage{amsopn}

\makeatletter
\newcommand*{\addFileDependency}[1]{
  \typeout{(#1)}
  \@addtofilelist{#1}
  \IfFileExists{#1}{}{\typeout{No file #1.}}
}
\makeatother

\newcommand*{\myexternaldocument}[1]{%
    \externaldocument{#1}%
    \addFileDependency{#1.tex}%
    \addFileDependency{#1.aux}%
}

\ifpdf
\hypersetup{
  pdftitle={Homogenisation of nonlinear blood flow in periodic networks: the limit of small haematocrit heterogeneity},
  pdfauthor={Yaron Ben-Ami, Brian D. Wood, Joe M. Pitt-Francis, Philip K. Maini, and Helen M. Byrne}
}
\fi

\myexternaldocument{ex_supplement}

\begin{document}

\nolinenumbers

\maketitle

\begin{abstract}
In this work we develop a homogenisation methodology to upscale mathematical descriptions of microcirculatory blood flow from the microscale (where individual vessels are resolved) to the macroscopic (or tissue) scale. 
Due to the assumed two-phase nature of blood and specific features of red blood cells (RBCs), mathematical models for blood flow in the microcirculation are highly nonlinear, coupling the flow and RBC concentrations (haematocrit). In contrast to previous works which accomplished blood-flow homogenisation by assuming that the haematocrit level remains constant, here we allow for spatial heterogeneity in the haematocrit concentration and thus begin with a nonlinear microscale model. We simplify the analysis by considering the limit of small haematocrit heterogeneity which prevails when variations in haematocrit concentration between neighbouring vessels are small.
Homogenisation results in a system of coupled, nonlinear partial differential equations describing the flow and haematocrit transport at the macroscale, in which a nonlinear Darcy-type model relates the flow and pressure gradient via a haematocrit-dependent permeability tensor.  During the analysis we obtain further that haematocrit transport at the macroscale is governed by a purely advective equation. 
Applying the theory to particular examples of two- and three-dimensional geometries of periodic networks, we calculate the effective permeability tensor associated with blood flow in these vascular networks. We demonstrate how the statistical \linebreak
distribution of vessel lengths and diameters, together with the average haematocrit level, affect the statistical properties of the macroscopic permeability tensor. These data can be used to simulate blood flow and haematocrit transport at the macroscale.
\end{abstract}



\section{Introduction}
The microcirculation is a complex network of vessels whose
primary function is to distribute oxygen and nutrients to the tissues and remove waste products. Due to the large number of vessels and their characteristically small diameter, the microcirculation also plays a major role in blood flow resistance. In diseased states, increases in blood-flow resistance can lead to increased blood pressure and risk of cardiovascular events such as stroke \cite{letcher1981direct,lowe1997blood}; additionally, the abnormal structure of the microcirculatory network in tumours, characterised by many immature leaky vessels and small inter-branching distances, leads to inefficient and inhomogeneous
distribution of oxygen, which can manifest as tumour hypoxia \cite{Jain05,nagy2009tumour}, and have poor consequences in terms of tumour prognosis.
Therefore, it is of evident interest to gain a better understanding of how the structure of the microcirculatory network affects blood flow. 

State-of-the-art imaging modalities (e.g., \cite{brown2022quantification,d2018computational}) can now generate detailed \linebreak
information about microcirculatory networks with high fidelity and spatial resolution. However numerical simulations of blood flow and the associated transport of oxygen at the scale of entire tissues is still too computationally intensive to conduct, even for relatively small volumes of tissue. This is especially true when the intrinsic nonlinearities of blood flow (see below) are taken into consideration and when the flow is unsteady.

One solution to the problem of resolution over large volumes of tissue is to \emph{replace} the microscale mass and momentum balances with an upscaled representation. This is particularly effective when microscale information is not critical for the modeling goals. One such methodology that has been applied successfully in a variety of fields across science and engineering is the method of multiple-scale homogenisation.  A useful feature of this approach is that the microstructural properties of the vascular network are related to the macroscale transport properties in a coherent and \linebreak predictable way. Thus, homogenisation of blood flow in vascular networks can leverage the (nowadays) available microscale data of network structures (e.g., \cite{brown2022quantification,d2018computational}) to generate predictions regarding measurable distributions of flow and oxygen transport at the macroscale (see for example, \cite{gonccalves2015decomposition,matsumoto2010imaging}).

Within the homogenisation framework, a representative volume is considered such that a separation of length scales is achieved.  In other words, the characteristic length scale of the representative volume is sufficiently smaller than the macrosopic length scale (defining, for example, the tissue comprising an organ).  While vascular networks do not always allow definition of a \emph{representative} volume, at the scale of the network of the smallest microcapillaries, this approximation is reasonable due to the homogeneous, space-filling nature of the capillary bed \cite{lorthois2010fractal}. 
The equations governing the microscale problem in the representative control volume are asymptotically \linebreak expanded in powers of the ratio of the micro-to-macro length scales. This procedure results in solvability conditions in terms of governing equations for the propagation of the leading-order solution at the macroscale.
In the specific context of low-Reynolds-number, fixed-haematocrit flow  in networks, the homogenisation process results in a Darcy-type equation relating the flow and pressure gradient through an effective permeability tensor which depends on the geometry of the representative microscale network \cite{chapman2021}. 

Blood flow in the microcirculation (specifically in the small-arteriolar and \linebreak capillary regions) is slow ($U \sim 1\, \unit{mm/s}$), the characteristic diameter is small ($D \sim 10\, \unit{\mu m}$), and a typical value of the kinematic viscosity of blood is $\nu \sim 1\,\unit{mm^2/s}$. Consequently, the Reynolds number ($\mathrm{Re}=UD/\nu$), describing the ratio of viscous and inertial terms in the momentum balance, is very small ($\mathrm{Re} \sim 0.01 \ll 1$). For simplicity, blood is often viewed as a Newtonian fluid which, together with the low-Reynolds dynamics, justifies the use of a linear, Stokes-flow model. However, this approximation ignores some of the complexities arising from the particulate nature of blood; red blood cells (RBC) are deformable particles that occupy a significant proportion of the blood volume (typically $\sim 45 \%$ in large vessels, but can be somewhat lower in the microcirculation \cite{boyle1988microcirculatory}). The high elasticity of the RBCs, coupled with the shear forces induced by the plasma flow, affect RBC migration in several ways, the most important of which is their tendency to migrate transversely, away from the vessel wall, leaving a cell-free layer in the vicinity of the wall. This transverse migration alters the resistance to blood flow as the haematocrit volume fraction and vessel diameter change, a phenomenon known as the F\aa{}hr\ae{}us-Lindqvist effect \cite{FL31}. To account for this effect in continuum models for blood flow, Pries et al. \cite{pries1994resistance} derived an empirical model for the apparent viscosity of blood as a function of the discharge haematocrit (ratio of RBC flux to total flux) and vessel diameter.
The existence of an RBC-depleted layer in the vicinity of the vessel wall introduces an additional complexity of blood flow -- a nonlinear biasing of the haematocrit flow at a vessel bifurcation towards the daughter branch with the higher flow, a phenomenon known as ``plasma skimming'', which leads to an inhomogeneous distribution of haematocrit in the vascular network \cite{pries1989red}.
Together, the effects of the haematocrit-dependent viscosity and biased haematocrit splitting at vessel bifurcations result in a highly nonlinear system of equations coupling haematocrit concentrations and flow rates in different vessels of a network. 

The aforementioned nonlinearities have been shown to drive the emergence of multiple equilibria and oscillatory dynamics in irregularly-structured vessel networks \cite{ben2022structural, karst2023modeling,KS15}. Many previous papers on the topic of homogenisation of blood-flow (e.g., \cite{chapman2008multiscale,el2015multi,peyrounette2018multiscale,shipley2010multiscale,shipley2020hybrid,shipley2020four}), assumed a constant haematocrit so that the nonlinearities associated with spatially varying haematocrit were neglected. In this work we maintain the spatial variation of haematocrit, and account for the resulting nonlinear properties of blood in our homogenisation methodology.  This approach  requires the momentum and haematocrit balances to be coupled. We found that if vessel diameters are sufficiently large ($\gg 1\,\unit{\mu m}$), the splitting phenomenon at vessel bifurcations is subdominant such that, at leading order, the average-macroscopic haematocrit is merely advected by the average flow. In turn, the uniformity of haematocrit at the \emph{microscale} (it changes only on the \emph{macroscopic} length scale) results in a nonlinear Darcy relation between the macroscale flow and pressure gradient, where the permeability tensor depends on the macroscopic haematocrit.

The remainder of the manuscript is structured as follows. In \cref{sec:blood_flow_model} we describe the model for blood flow in vascular networks. We present the homogenisation methodology in two parts: in \cref{sec:homog_flow} we extend the homogenisation procedure \linebreak developed by \chap for homogenisation of flow in networks to account for vessel conductances that are haematocrit-dependent. This yields a Darcy-type equation in which the permeability tensor is a function of the haematocrit. Then, in \cref{sec:homog_haematocrit} we close the system of equations by introducing our approach for homogenisation of haematocrit propagation which is based on an assumption of small haematocrit heterogeneity (small deviations from uniform haematocrit).
The resulting macroscale homogenised equations include the permeability tensor, which contains details about the structure of the microscale representative network. In \cref{sec:permeability} we present results for the permeability tensor as a function of the haematocrit level and the geometrical and topological features of the elementary network (i.e., the representative volume). The paper concludes in \cref{sec:conclusion} where we summarise our findings and outline possible directions for future research.

\section{Model for blood flow in vascular networks\label{sec:blood_flow_model}}

\subsection{Model for blood flow in a single vessel\label{subsec:single_vessel}}
In this section we explain how we relate the blood flow through a single cylindrical vessel to the pressure difference between its ends (nodes). We denote by $p^*_i$ the pressure at node $i$ and by $q^*_{ij}$ and $h_{ij}$ the flow rate and haematocrit (the ratio of RBC flow to the total flow), respectively, in the vessel connecting nodes $i$ and $j$.
Here, and throughout the manuscript, asterisks denote dimensional quantities.

We begin by describing the flow in a single vessel segment.  Thus, for each vessel segment, $ij$, conservation of fluid mass yields the relation
\begin{equation}
\frac{\partial q_{ij}^*}{\partial x^*}=0.
\label{eq:incom}
\end{equation}
Here, $x^*$ is the coordinate coincident with  the vessel axis. Consequently, we have that $q_{ij}^*$ is uniform along each vessel segment (but may vary between segments). Our approach requires that we know both the distribution of pressures among segment nodes, and the haematocrit in each segment.  As standard for haematocrit, we assume that it is transported hyperbolically, which is equivalent to assuming that the Peclet number is sufficiently large.  Thus, assuming a steady-state flow, conservation of haematocrit is given by
\begin{equation}
\frac{\partial }{\partial x^*}(q_{ij}^*h_{ij})=0,
\label{eq:h_conserv}
\end{equation}
where $q_{ij}^*h_{ij}$ represents the RBC flux (and where $0 \le h_{ij} \le 1$ represents the RBC concentration). By combining \cref{eq:incom} and \cref{eq:h_conserv} it is immediately apparent that $h_{ij}$ does not vary with the position in a given vessel.

Since we focus on blood flow in the microcirculation, we assume low-Reynolds-number flow and that the pulsatile variations in flow are sufficiently damped that steady-state flow conditions prevail.  Each segment is assumed to be cylindrical, with a uniform diameter, $d_{ij}^{*}$, whose length, $l_{ij}^{*}$, is much larger than its diameter. Under these assumptions, the flow satisfies Poiseuille's law
\begin{equation}
q_{ij}^{*}=-\frac{\pi}{128}\frac{d_{ij}^{*4}}{l_{ij}^{*}\mu_{0}^{*}\mu(h_{ij},\widetilde{d}_{ij})}(p^*_{j}-p^*_{i}).
\label{Poiseuille}
\end{equation}
Here, the apparent viscosity of blood is given
by $\mu_{0}^{*}\mu$, where $\mu_{0}^{*}$ represents the plasma
viscosity and $\mu$ is a nondimensional function which represents the effect that the haematocrit fraction, $h_{ij}$, and vessel diameter ($\widetilde{d}_{ij}=d^*_{ij}/d_{\mu}^{*}$ denotes the dimensionless 
vessel diameter, with $d_{\mu}^{*}\equiv 1~ \unit{\mu m}$) have on the effective blood viscosity. The vessel diameter is scaled by  $d_{\mu}^{*}$ in the viscosity function for consistency with the functional form proposed by Pries et al. \cite{pries1994resistance} to model the empirical blood viscosity. The explicit expression for the blood viscosity used in this work is detailed in \cref{app:blood_viscosity} and is based on the model of \cite{pries1994resistance}.

Upon translating the balances to dimensionless form, we scale diameters of vessels by $d_{0}^{*}$, a representative vessel diameter in the network, which is typically much larger than $1~ \unit{\mu m}$ ($d_{0}^{*} \gg d_{\mu}^{*}$); and vessel lengths by a reference length, $l_0^*$. Additionally, we introduce
the macroscopic length scale, $\Lambda_{0}^{*}$, representing the characteristic size of the full vascular network (i.e., a network containing a large number of vessels), such that we can define
a small parameter $\epsilon$ as follows
\begin{equation}
\epsilon=\frac{l_{0}^{*}}{\Lambda_{0}^{*}}\ll1.
\label{eq:epsilon}
\end{equation}
The pressure is normalised by the macroscopic pressure difference,
$\triangle p_{0}^{*}$, representing the characteristic magnitude of the pressure difference between the macroscopic \linebreak boundaries. Accordingly, we scale the flow by $\Gamma_0 \equiv\pi d_{0}^{4*}\triangle p_{0}^{*}/(128\Lambda_{0}^{*}\mu_{0}^{*})$
to obtain the nondimensional equation

\begin{equation}
q_{ij}=-\frac{g_{ij}}{\epsilon}(p_{j}-p_{i}),\label{eq:q_p}
\end{equation}
where the nondimensional conductance, $g_{ij}$, is given by
\begin{equation}
g_{ij}=\frac{d_{ij}^{4}}{l_{ij}\mu(h_{ij},\widetilde{d}_{ij})}.\label{eq:g_ij_nondimensional}
\end{equation}

\subsection{Model for blood flow in a periodic network\label{subsec:network_model}}
In this work, blood flow is represented as occurring in a  network of vessel segments (edges) that are connected at bifurcations (nodes). For closure, we assume periodicity of the network such that it forms a tessellation of periodic unit cells filling the macroscale domain. Each unit cell is composed of $N$ nodes, connected \emph{internally} by edges, and is \emph{inter-connected} to a number of identical neighbouring unit cells. The (spatially uniform) set of relative displacements between identical nodes in neighbouring unit cells defines the periodic pattern of the network. Following \cite{chapman2021}, we denote the relative displacement vector by $\mathbf{r} \circ \mathbf{L}$, where $\mathbf{L}$ is the vector of unit-cell dimensions and $\circ$ marks the Hadamard product; $\bf r$ indicates the direction of displacement in cell-length units.  

In periodic networks, two nodes in neighbouring unit cells may be connected via either ``inlet'' or ``outlet'' vessels.  Within each unit cell, there are exactly $N$ nodes.  Any such node may be connected to any of the other $N-1$ nodes within the unit cell, or to nodes that belong to adjacent unit cells.  Thus, in order to fully describe the periodic unit cell, the vessels spanning more than one unit cell must be appropriately identified. For that purpose we follow \cite{chapman2021} and denote each vessel by a double-small-letter-subscript, $ij$, marking its end nodes, together with a superscript $\mathbf{r}$, where $\mathbf{r}$ represents a displacement vector specifying an adjacent unit cell. Accordingly, the set $\mathcal{N} = \{\boldsymbol{0}, \mathbf{r}_1,...,\mathbf{r}_{N_\mathrm{adj}}\}$ defines the set of edge displacements, where $N_\mathrm{adj}$ is the number of neighbouring unit cells. Here, $\mathbf{r}=\boldsymbol{0}$ will mark a vessel whose coordinates $i$ and $j$ are both within a single unit cell; the remaining $N_\mathrm{adj}$ elements of $\mathcal{N}$ specify the displacements for the remaining edges that begin with one of the $N$ nodes in the unit cell but terminate in an adjacent unit cell. \Cref{fig:periodic_network_illustration} shows a schematic of a periodic network, where some examples for the convention used to define the different edges are shown.

\begin{figure}[H]
\includegraphics[scale=0.44]{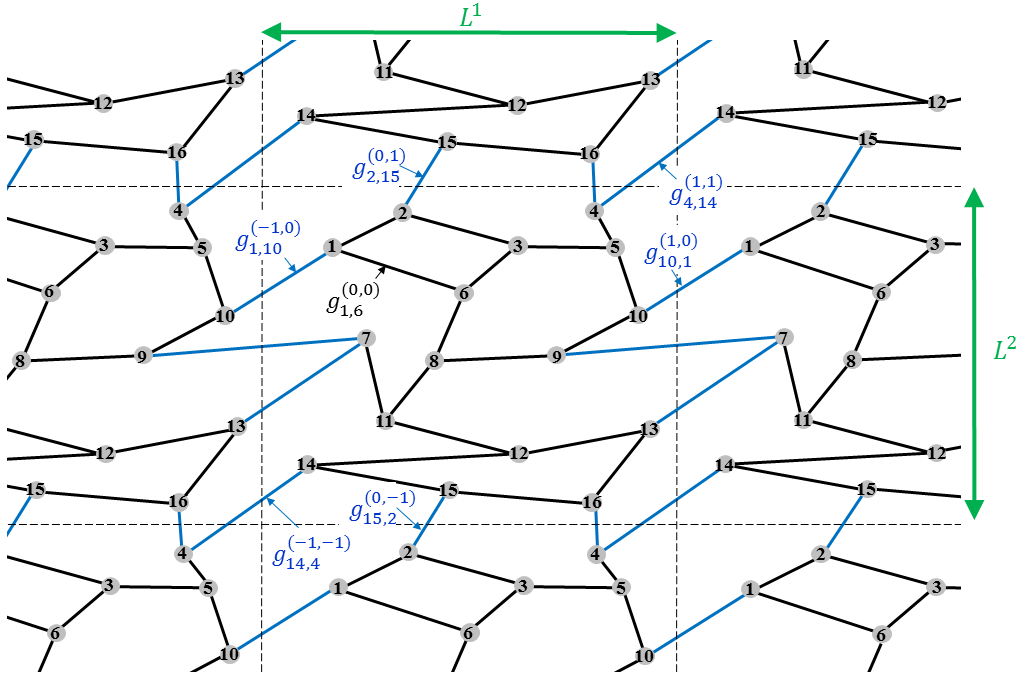}
\centering{}\caption{Schematic of a 2D periodic network. Each unit cell in this example network is composed of 16 nodes, 18 cell-internal edges (black), and 12 edges spanning more than one unit cell (blue). In this specific example, $\mathbf{L}=(L^1,L^2)$ and $\mathcal{N} = \{(0,0),(1,0),(-1,0),(0,1),(0,-1),(1,1),(-1,-1)\}$. Some examples for the convention used to describe the edge conductances, $g_{ij}^{\bf r}$, are shown in the figure.
\label{fig:periodic_network_illustration}}
\end{figure}

For simplicity of notation, we will sometimes use a single capital-letter-subscript, $J$, to denote a vessel so that the transformation between vessel index and its \linebreak
corresponding double-node indices is unique. For example, the haematocrit
in vessel $J$ can also be represented as the haematocrit in vessel
$j_{0}j_{1}$ where $j_{0}$ and $j_{1}$ are the end nodes of vessel
$J$, $h_{J} \equiv h_{j_{0}j_{1}}^{\mathbf{r}}$. Its periodic counterpart will be denoted by $h_{\hat{J}} \equiv h_{j_{1}j_{0}}^{-\mathbf{r}}$.

At this juncture, we specify the system of equations governing the flow in the unit cell. By applying conservation of flow at each node, $i$, we write
\begin{equation}
\sum_{\mathbf{r} \in \mathcal{N}}\sum_{j=1}^{N}q_{ij}^{\mathbf{r}}=\sum_{\mathbf{r}}\sum_{j=1}^{N}g_{ij}^{\mathbf{r}}(p_{j}-p_{i})=0,\;\;\;\text{for}\;i=1,..,N,
\label{eq:sum_q}
\end{equation}
where the summation over all $\mathbf{r} \in \mathcal{N}$ covering both cell-internal vessels and vessels connected to each of the neighbouring unit cells. Here, $g_{ij}^{\mathbf{r}} \equiv 0$ if nodes $i$ and $j$ are not connected under the orientation defined by~$\mathbf{r}$.

In order to close equations (\ref{eq:sum_q}), the conductance, $g_{ij}^{\mathbf{r}}=g_{ij}^{\mathbf{r}}(h_{ij}^{\mathbf{r}})$, must be \linebreak evaluated. Accordingly, the distribution of haematocrit in the vessel network must be determined.
In general, the haematocrit in a vessel depends on the haematocrit and flow in its parent vessel(s) \cite{pries1989red}.
Therefore, three different types of vessel exist in each unit cell, as illustrated in \cref{fig:vessel_types}: (i) vessels emerging from bifurcating flow junctions;
(ii) vessels originating from converging flow junctions; (iii) vessels originating from nodes in neighbouring unit cells (``inlet vessels''). Here, we define the haematocrit in each vessel as
\begin{equation}
h_{J}=F_{J}\left( h_{J_p}, q_{J_p} \right),\;\;\;\text{for}\;J=1,..,M
\label{eq:h_general}
\end{equation}
where $J_p$ is (are) the parent vessel(s) of vessel $J$, and $F_{J}$ depends on the type of the vessel; we postpone its definition to \cref{sec:homog_haematocrit}. The number of vessels in the unit cell, $M$, refers to all the vessels that are connected to at least a single node in the unit cell.

\begin{figure}[H]
\includegraphics[scale=0.45]{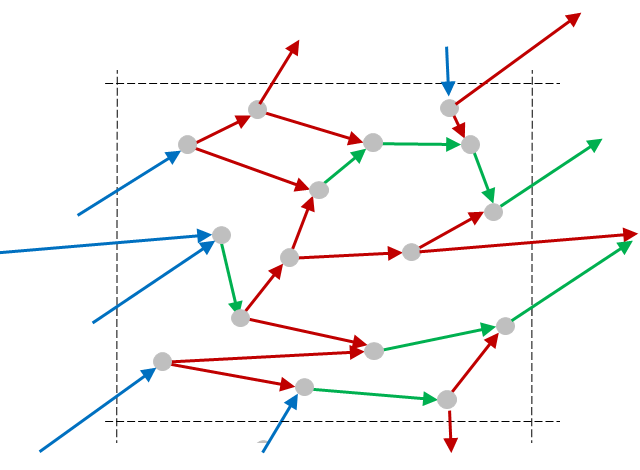}
\centering{}\caption{Examples of different vessel types in a periodic unit cell of vascular networks. In this example, the direction of flow at each edge of the unit cell from \cref{fig:periodic_network_illustration} was chosen arbitrarily and was denoted by the direction of the arrow. The different vessel types were then marked accordingly: (i) vessels emerging from bifurcating-flow junctions (red); (ii) vessels originating from converging-flow junctions (green); (iii) vessels originating from nodes in neighbouring unit cells (blue).
\label{fig:vessel_types}}
\end{figure}

\Cref{eq:sum_q,eq:h_general}, together with \cref{eq:q_p,eq:g_ij_nondimensional}, constitute a coupled system of nonlinear equations for $p_i$ ($i=1,...,N$), $q_J$, and $h_J$ ($J=1,...,M$).
In this paper we seek to upscale the microscopic unit-cell-level description given by \cref{eq:q_p,eq:g_ij_nondimensional,eq:sum_q,eq:h_general} to the macroscopic level using the method of multiscale homogenisation. The methodology of homogenisation of flow in periodic networks was recently formulated by \chap for the linear case where the vessel conductances are independent of the flow. 
Here, we generalise the methodology of \chap to a case where the conductances are determined by a viscosity that is \emph{nonlinearly} coupled to the vessel haematocrit; the haematocrit, in turn, depends on the flow through interactions at vessel bifurcations.  In \cref{sec:homog_flow} we will follow the scheme of \chap to derive macroscopic equations for the flow in terms of the 
conductances, $g_{ij}^{\mathbf{r}}$. Then, in \cref{sec:homog_haematocrit} we will
derive homogenized equations governing haematocrit propagation through the network so that the conductance can be coupled to the haematocrit through a nonlinear apparent viscosity function.

\section{Homogenisation of the flow} \label{sec:homog_flow}

\subsection{Asymptotic solution for the pressure}

One critical difference in this problem from the work conducted by \chap is that the \linebreak conductances, $g^{\bf r}_{ij}$, depend upon spatial location through the coupling with the \linebreak haematocrit (and, hence, ultimately the pressure).  While the conductance in each vessel \emph{segment} is constant, the conductance may be different for each such segment.  To spatially identify vessel segments, each vessel segment is tagged with the coordinate representing its midpoint.  Thus, we have

\begin{equation}
    g^{\bf r}_{ij} = g^{\bf r}_{ij}\hspace{-1mm}\left(\mathbf{x}+\tfrac{1}{2}\epsilon(\mathbf{x}_{i}+\mathbf{x}_{j}+\mathbf{r}\circ\mathbf{L}) \right).
\end{equation}
Here, $\mathbf{x}$ is the origin of the unit cell normalized by the macroscopic length scale, $\Lambda_0$. The dimensional position of a node within the unit cell, as well as the length of the unit-cell, both scale with the characteristic vessel length, such that
\[
\mathbf{x}_{i}^* =  l_0 \mathbf{x}_{i}\ \ \ \text{and}\ \ \ \mathbf{L}^* =  l_0 \mathbf{L}.
\]
Therefore, the normalization by $\Lambda_0$ yields the $O(\epsilon)$ scaling.
For notational \linebreak convenience, we set
\begin{equation}
    \overline{\mathbf{x}}^{\bf r}_{ij} = \mathbf{x}+\tfrac{1}{2}\epsilon(\mathbf{x}_{i}+\mathbf{x}_{j}+\mathbf{r}\circ\mathbf{L}).
\label{eq:x_ij}
\end{equation}

We are now ready to write the nodal pressures
in \cref{eq:sum_q} as a function of their spatial locations as follows
\begin{equation}
\sum_{\mathbf{r} \in \mathcal{N}}\sum_{j=1}^{N}g_{ij}^{\mathbf{r}}(\overline{\mathbf{x}}^{\bf r}_{ij})\Big[p(\mathbf{x}+\epsilon\mathbf{x}_{j}+\epsilon\mathbf{r}\circ\mathbf{L})-p(\mathbf{x}+\epsilon\mathbf{x}_{i})\Big]=0,\;\;\;\text{for}\;i=1,..,N.\label{eq:flow_balance}
\end{equation}
\rule[0.5ex]{1\columnwidth}{0.5pt}

\subsubsection*{Sidebar: Expansion of the flow fields and network periodicity\label{sec:sidebar}}

At this juncture, we would like to 
perform Taylor series expansions for
the pressure and conductance 
in a neighbourhood of
the cell origin, $\bf x$. 
First, it is important to explain and justify the use of a Taylor series expansion for the conductance which is generally not continuous at vessel bifurcations (also, while the pressure is continuous, its gradient is not). Here, we make the case for the conductance, but the same arguments apply for the pressure.

Consider $g_{ij}^{\mathbf{r}} \equiv \bf G$, a vector of conductances of all vessels in a single unit cell. Then, expansion of $\bf G$ around the cell origin means that the change of a specific vessel conductance (a single component in $\bf G$) between unit cells is continuous. As such, the discontinuity of the conductance at vessel bifurcations (which manifests by $G_I \neq G_J$, where $I$ and $J$ are two vessels connected at a bifurcation) is not important for the expansion, because we expand each component separately. 

The above discussion leads to the explanation of how we impose network \linebreak \emph{periodicity}: without loss of generality, we can consider two neighbouring unit cells, which are concatenated in the $\bf x^1$ direction and numbered $n$ and $n+1$; the unit cells have a cell-length, $\epsilon L^1$ in the $\bf x^1$ direction. We accordingly define the $x^1$-component of the origin of the $n$th cell as $(n-1)\epsilon L$. Then, the conductance of a vessel originating in cell $n$ and terminating in cell $n+1$ (having an indicator $\mathbf{r^+}\equiv(1,0,0)$) is
\begin{equation}
\begin{aligned}
g_{ij}^{\mathbf{r}^+}\left(n,\overline{\mathbf{x}}^{\mathbf{r^+}}_{ij}\right)= & g\left((n-1)\epsilon L \mathbf{x^1} + \tfrac{1}{2}\epsilon(\mathbf{x}_{i}+\mathbf{x}_{j}+L\mathbf{x^1})\right) \\
= & g\left(\epsilon\left[ (n-\tfrac{1}{2})L\mathbf{x^1} + \tfrac{1}{2}(\mathbf{x}_{i}+\mathbf{x}_{j})\right] \right),
\end{aligned}
\label{eq:g_r_1}
\end{equation}
where $g(\mathbf{x})$ is a continuous and once-differentiable function of the spatial coordinate, $\mathbf{x}$. Equivalently, the conductance of the \emph{same} vessel can be written as a vessel originating in cell $n+1$ and terminating in cell $n$ (with an indicator $\mathbf{r^-}=-\mathbf{r^+}$):
\begin{equation}
\begin{aligned}
 g_{ji}^{\mathbf{r^-}}\left(n+1,\overline{\mathbf{x}}^{\mathbf{r^-}}_{ji}\right)=& g\left(n\epsilon L\mathbf{x^1} + \tfrac{1}{2}\epsilon(\mathbf{x}_{i}+\mathbf{x}_{j}-L\mathbf{x^1})\right) \\
 = &  g\left(\epsilon\left[ (n-\tfrac{1}{2})L\mathbf{x^1} + \tfrac{1}{2}(\mathbf{x}_{i}+\mathbf{x}_{j})\right] \right),
 \end{aligned}
 \label{eq:g_r_minus_1}
\end{equation}
which is identical to the expression in \cref{eq:g_r_1}. Therefore, the \emph{periodicity} condition is given by
\begin{equation}
g_{ij}^\mathbf{r^+}\left(n,\overline{\mathbf{x}}^{\mathbf{r^+}}_{ij}\right) = g_{ji}^{\mathbf{r^-}}\left(n+1,\overline{\mathbf{x}}^{\mathbf{r^-}}_{ji}\right).
\label{eq:periodicity_n}
\end{equation}
Noting that the number of the unit cell can be simply replaced by the spatial location of its origin, we can write the periodicity condition \cref{eq:periodicity_n} in the form
\begin{equation}
g_{ij}^{\mathbf{r}}(\mathbf{x})=g_{ji}^{-\mathbf{r}}(\mathbf{x}),
\label{eq:g_period}
\end{equation}
such that any variation in the conductance of a specific vessel, $ij$, between different cells is accounted for via 
the dependence on the spatial location, $\bf x$.
\newline
\rule[0.5ex]{1\columnwidth}{0.5pt}

\vspace{3mm}
Returning now to \cref{eq:flow_balance} and expanding it in a Taylor series centred on 
the cell origin, ${\bf x}$, the result is a series in powers of $\epsilon$ 
\begin{equation}
\begin{aligned} 
& \sum_{\mathbf{r} \in \mathcal{N}}\sum_{j=1}^{N}\left[g_{ij}^{\mathbf{r}}+\frac{1}{2}\epsilon\sum_{s=1}^{3}(x_{i}^{s}+x_{j}^{s}+r^{s}L^{s})\frac{\partial g_{ij}^{\mathbf{r}}}{\partial x^{s}}+...\right] \\
& \times\Biggl[\Biggr. p_{j}-p_{i}+\epsilon\sum_{k=1}^{3}(x_{j}^{k}+r^{k}L^{k})\frac{\partial p_{j}}{\partial x^{k}}-\epsilon\sum_{k=1}^{3}x_{i}^{k}\frac{\partial p_{i}}{\partial x^{k}} \\
& +\frac{\epsilon^{2}}{2}\sum_{k=1}^{3}\sum_{l=1}^{3}(x_{j}^{k}+r^{k}L^{k})(x_{j}^{l}+r^{l}L^{l})\frac{\partial^{2}p_{j}}{\partial x^{k}\partial x^{l}}
-\frac{\epsilon^{2}}{2}\sum_{k=1}^{3}\sum_{l=1}^{3}x_{i}^{k}x_{i}^{l}\frac{\partial^{2}p_{i}}{\partial x^{k}\partial x^{l}}+...\Biggl.\Biggr] = 0.
\end{aligned}\label{eq:expansion}
\end{equation}
The expansion of $g_{ij}^{\mathbf{r}}$
reflects the dependence of the conductance on the haematocrit; later, this will be expanded explicitly in terms of the haematocrit. From here forward, unless stated otherwise, functions are evaluated at the cell origin, $\mathbf{x}$.

To clarify the notation,  the nodal pressures are written as a vector, \linebreak
$\mathbf{P}=(p_{1},...,p_{N})^{T}$. Then, \cref{eq:expansion} can be expressed as a system of $N$ equations in the following compact matrix form
\begin{equation}
A_{0}\mathbf{P}+\epsilon\left[\sum_{l=1}^{3}B_{0}^{l}\frac{\partial\mathbf{P}}{\partial x^{l}}+A_{1}\mathbf{P}\right]+\epsilon^{2}\left[\sum_{k=1}^{3}\sum_{l=1}^{3}C_{0}^{kl}\frac{\partial^{2}\mathbf{P}}{\partial x^{k}\partial x^{l}}+\sum_{l=1}^{3}B_{1}^{l}\frac{\partial\mathbf{P}}{\partial x^{l}}\right]=O(\epsilon^{3}),\label{eq:eq_P_vec}
\end{equation}
where
\[
\begin{gathered}
A_{0}=\sum_{\mathbf{r} \in \mathcal{N}}\left[g_{ij}^{\mathbf{r}}-\delta_{ij}\sum_{m=1}^{N}g_{im}^{\mathbf{r}}\right],\\
B_{0}^{l}=\sum_{\mathbf{r} \in \mathcal{N}}\left[(x_{j}^{l}+r^{l}L^{l})g_{ij}^{\mathbf{r}}-\delta_{ij}x_{i}^{l}\sum_{m=1}^{N}g_{im}^{\mathbf{r}}\right],\\
A_{1}=\frac{1}{2}\sum_{\mathbf{r} \in \mathcal{N}}\left[\sum_{s=1}^{3}(x_{i}^{s}+x_{j}^{s}+r^{s}L^{s})\frac{\partial g_{ij}^{\mathbf{r}}}{\partial x^{s}}-\delta_{ij}\sum_{m=1}^{N}\sum_{s=1}^{3}(x_{i}^{s}+x_{j}^{s}+r^{s}L^{s})\frac{\partial g_{im}^{\mathbf{r}}}{\partial x^{s}}\right],\\
C_{0}^{kl}=\frac{1}{2}\sum_{\mathbf{r} \in \mathcal{N}}\left[(x_{j}^{k}+r^{k}L^{k})(x_{j}^{l}+r^{l}L^{l})g_{ij}^{\mathbf{r}}-\delta_{ij}x_{i}^{k}x_{i}^{l}\sum_{m=1}^{N}g_{im}^{\mathbf{r}}\right],
\end{gathered}
\]
and
\[
\begin{aligned}
   B_{1}^{l} = & \sum_{\mathbf{r} \in \mathcal{N}}\Biggl[\Biggr.(x_{j}^{l}+r^{l}L^{l})\sum_{s=1}^{3}(x_{i}^{s}+x_{j}^{s}+r^{s}L^{s})\frac{\partial g_{ij}^{\mathbf{r}}}{\partial x^{s}} \\
  & \qquad \qquad \qquad \qquad \qquad \qquad -\delta_{ij}x_{i}^{l}\sum_{m=1}^{N}\sum_{s=1}^{3}(x_{i}^{s}+x_{m}^{s}+r^{s}L^{s})\frac{\partial g_{im}^{\mathbf{r}}}{\partial x^{s}}\Biggl.\Biggr]. 
\end{aligned}
\]
In these expressions, $\delta_{ij}$ is the conventional Kronecker delta tensor.  The structure of the matricies $A_0$, $A_1$, $B_1^l$ and $C_0^{kl}$ depend upon the particular connectivity of the network. 

We aim to solve \cref{eq:eq_P_vec} in the limit $\epsilon\rightarrow0$
by seeking an asymptotic expansion for $\mathbf{P}$ of the form
\begin{equation}
\mathbf{P}=\mathbf{P}^{(0)}+\epsilon\mathbf{P}^{(1)}+\epsilon^{2}\mathbf{P}^{(2)}+...\;.\label{eq:p_exp}
\end{equation}
At this juncture, the conductances are expanded in the form
\begin{equation}
g_{ij}^{\mathbf{r}}=g_{ij}^{\mathbf{r}(0)}+\epsilon g_{ij}^{\mathbf{r}(1)}+...\,.
\label{eq:g_ij_exp}
\end{equation}
The expansion of the conductances will ultimately correspond to the expansion of the haematocrit; this will be explained in further detail in \cref{sec:homog_flow} [see \cref{eq:expansion_h_q} \emph{et seq.}].

\vspace{2mm}

\paragraph*{Identities}

Defining the $[N \times 1]$ vector $\mathbf{u}=(1,...,1)^{T}$, we have the following identities (these were proved in \cite{chapman2021}; for completeness, we reiterate the proofs in \cref{app:identities_1}):

\begin{subequations}\label{eq:identities}
\begin{equation}
A_{0} = A_{0}^{T}; \label{eq:identities:a}
\end{equation}
\begin{equation}
A_{0}\mathbf{u} = A_{1}\mathbf{u}=\mathbf{0};
\label{eq:identities:b}
\end{equation}
\begin{equation}
\mathbf{u}^{T}A_{0} =\mathbf{0}; \label{eq:identities:c}
\end{equation}
\begin{equation}
\mathbf{u}^{T}B_{0}^{k}\mathbf{u} =0. \label{eq:identities:d}
\end{equation}
\end{subequations}

Substituting from \cref{eq:p_exp,eq:g_ij_exp} into \cref{eq:eq_P_vec}, and equating the $O(1)$ terms, we have
\begin{equation}
A_{0}(g_{ij}^{\mathbf{r}(0)})\mathbf{P}^{(0)}=0.
\label{eq:zero_order}
\end{equation}
Using \cref{eq:identities:b} we can write
\begin{equation}
\mathbf{P}^{(0)}=p(\mathbf{x})\mathbf{u},\label{eq:p_0}
\end{equation}
showing that, to leading order, the pressure is uniform across the unit cell.

Substituting from \cref{eq:p_0} into \cref{eq:eq_P_vec} and equating terms
of $O(\epsilon)$ we have
\begin{equation}
   A_{0}(g_{ij}^{\mathbf{r}(0)})\mathbf{P}^{(1)}+\left(A_{0}(g_{ij}^{\mathbf{r}(1)})+A_{1}(g_{ij}^{\mathbf{r}(0)})\right)p\mathbf{u}+\sum_{l=1}^{3}B_{0}^{l}(g_{ij}^{\mathbf{r}(0)})\frac{\partial p}{\partial x^{l}}\mathbf{u}  = 0.
\end{equation}
Then, using the identity in \cref{eq:identities:b} we have
\begin{equation}
  A_{0}(g_{ij}^{\mathbf{r}(0)})\mathbf{P}^{(1)}+\sum_{l=1}^{3}B^{l}(g_{ij}^{\mathbf{r}(0)})\frac{\partial p}{\partial x^{l}}\mathbf{u}  = 0.
\end{equation}
From \cref{eq:identities:a,eq:identities:b}, $A_{0}^{T}=A_{0}$
and $A_{0}\mathbf{u}=0$, such that
\begin{equation}
A_{0}^{T}(g_{ij}^{\mathbf{r}(0)})\mathbf{u}=\mathbf{0}
\label{eq:A_0_g_0_u}
\end{equation}
By appealing to the Fredholm alternative, we deduce that since \cref{eq:A_0_g_0_u} holds there is a solution for $\mathbf{P}^{(1)}$ only if 
\begin{equation}
\mathbf{u}^{T}\sum_{l=1}^{3}B_{0}^{l}(g_{ij}^{\mathbf{r}(0)})\frac{\partial p}{\partial x^{l}}\mathbf{u}=0,
\label{eq:solv_P1}
\end{equation}
where \cref{eq:identities:d} guarantees that the this condition is satisfied for any $p$.
Then, we can write the $O(\epsilon)$ contribution to the pressure as
\begin{equation}
\mathbf{P}^{(1)}=T(\mathbf{x})\mathbf{u}+\sum_{l=1}^{3}\mathbf{w}^{l}\frac{\partial p}{\partial x^{l}},\label{eq:p_1}
\end{equation}
where the two terms on the right-hand side represent homogeneous and particular solutions, $T$
is an arbitrary function, and $\mathbf{w}^{l}$ is any vector satisfying
\begin{equation}
A_{0}(g_{ij}^{\mathbf{r}(0)})\mathbf{w}^{l}=-B_{0}^{l}(g_{ij}^{\mathbf{r}(0)})\mathbf{u}.\label{eq:w_sol}
\end{equation}
Substituting from \cref{eq:p_0,eq:p_1} into \cref{eq:eq_P_vec},
and equating terms of $O(\epsilon^{2})$, we have
\begin{equation}
\begin{aligned}
& A_{0}(g_{ij}^{\mathbf{r}(0)})\mathbf{P}^{(2)}+\left(A_{0}(g_{ij}^{\mathbf{r}(1)})+A_{1}(g_{ij}^{\mathbf{r}(0)})\right)\sum_{l=1}^{3}\mathbf{w}^{l}\frac{\partial p}{\partial x^{l}}\\
& +\sum_{k=1}^{3}B_{0}^{k}(g_{ij}^{\mathbf{r}(0)})\frac{\partial}{\partial x^{k}}\left(T(\mathbf{x})\mathbf{u}+\sum_{l=1}^{3}\mathbf{w}^{l}\frac{\partial p}{\partial x^{l}}\right) +\sum_{l=1}^{3}B_{0}^{l}(g_{ij}^{\mathbf{r}(1)})\frac{\partial p}{\partial x^{l}}\mathbf{u} \\
& +\sum_{k=1}^{3}\sum_{l=1}^{3}C_{0}^{kl}(g_{ij}^{\mathbf{r}(0)})\frac{\partial^{2}p}{\partial x^{k}\partial x^{l}}\mathbf{u}+\sum_{k=1}^{3}B_{1}^{k}(g_{ij}^{\mathbf{r}(0)})\frac{\partial p}{\partial x^{k}}\mathbf{u} = 0,
\end{aligned}
\end{equation}
where we have used the identity in \cref{eq:identities:b}. Then, by the Fredholm alternative, we obtain the following solvability condition
\begin{equation}
\begin{aligned}
& \mathbf{u}^{T}A_{1}(g_{ij}^{\mathbf{r}(0)})\sum_{l=1}^{3}\mathbf{w}^{l}\frac{\partial p}{\partial x^{l}}+\sum_{k=1}^{3}\sum_{l=1}^{3}\mathbf{u}^{T}B_{0}^{k}(g_{ij}^{\mathbf{r}(0)})\frac{\partial\mathbf{w}^{l}}{\partial x^{k}}\frac{\partial p}{\partial x^{l}} \\
& + \mathbf{u}^{T}\sum_{k=1}^{3}\sum_{l=1}^{3}C_{0}^{kl}(g_{ij}^{\mathbf{r}(0)})\frac{\partial^{2}p}{\partial x^{k}\partial x^{l}}\mathbf{u}+\mathbf{u}^{T}\sum_{l=1}^{3}B_{1}^{k}(g_{ij}^{\mathbf{r}(0)})\frac{\partial p}{\partial x^{l}}\mathbf{u}=0,
\end{aligned}
\label{eq:solve_p_O_2}
\end{equation}
where the dependence on $g_{ij}^{\mathbf{r}(1)}$ has been eliminated by using the identities in \cref{eq:identities:c,eq:identities:d}.

We define the symmetric tensor $\boldsymbol{\mathcal{K}}$
\begin{equation}
\mathcal{K}^{kl}=\mathbf{u}^{T}C_{0}^{kl}(g_{ij}^{\mathbf{r}(0)})\mathbf{u}+\frac{1}{2}\mathbf{u}^{T}B_{0}^{k}(g_{ij}^{\mathbf{r}(0)})\mathbf{w}^{l}+\frac{1}{2}\mathbf{u}^{T}B_{0}^{l}(g_{ij}^{\mathbf{r}(0)})\mathbf{w}^{k}\label{eq:K_tensor}
\end{equation}
and the antisymmetric tensor $\boldsymbol{\Omega}$
\begin{equation}
\Omega^{kl}=\frac{1}{2}\mathbf{u}^{T}B_{0}^{k}(g_{ij}^{\mathbf{r}(0)})\mathbf{w}^{l}-\frac{1}{2}\mathbf{u}^{T}B_{0}^{l}(g_{ij}^{\mathbf{r}(0)})\mathbf{w}^{k}.
\end{equation}
Additionally, we define the vector $\boldsymbol{\xi}$,
\begin{equation}
\xi^{l}=\sum_{k=1}^{3}\mathbf{u}^{T}B_{0}^{k}(g_{ij}^{\mathbf{r}(0)})\frac{\partial\mathbf{w}^{l}}{\partial x^{k}}+\mathbf{u}^{T}A_{1}(g_{ij}^{\mathbf{r}(0)})\mathbf{w}^{l}+\mathbf{u}^{T}B_{1}^{l}(g_{ij}^{\mathbf{r}(0)})\mathbf{u}.\label{eq:xi_vector}
\end{equation}
Then, \cref{eq:solve_p_O_2} can be written as
\begin{equation}
\mathcal{K}^{kl}\frac{\partial^{2}p}{\partial x^{k}\partial x^{l}}+\xi^{l}\frac{\partial p}{\partial x^{l}}=0,\label{eq:Darcy_modified_sym}
\end{equation}
where the terms $\Omega^{kl}\tfrac{\partial^{2}p}{\partial x^{k}\partial x^{l}}$
vanish since $\boldsymbol{\mathbf{\Omega}}$ is an antisymmetric tensor.
\vspace{2mm}
\paragraph*{Identities}

We can simplify the terms that appear in \cref{eq:K_tensor,eq:xi_vector} as follows (see \cref{app:identities_2} for details):
\begin{subequations}\label{eq:identities_2}
\begin{equation}
\mathbf{u}^{T}A_{1}(g_{ij}^{\mathbf{r}(0)})\mathbf{w}^{l} =  \sum_{\mathbf{r} \in \mathcal{N}}\sum_{i=1}^{N}\sum_{j=1}^{N}\sum_{s=1}^{3}r^{s}L^{s}w_{j}^{l}\frac{\partial g_{ij}^{\mathbf{r}(0)}}{\partial x^{s}}; \label{eq:identities_2:a}
\end{equation}
\begin{equation}
\mathbf{u}^{T}B_{1}^{l}(g_{ij}^{\mathbf{r}(0)})\mathbf{u}=\frac{1}{2}\sum_{\mathbf{r} \in \mathcal{N}}\sum_{i=1}^{N}\sum_{j=1}^{N}\sum_{s=1}^{3}\left[r^{l}L^{l}r^{s}L^{s}+2x_{j}^{l}r^{s}L^{s}\right]\frac{\partial g_{ij}^{\mathbf{r}(0)}}{\partial x^{s}};
\label{eq:identities_2:b}
\end{equation}
\begin{equation}
\mathbf{u}^{T}C_{0}^{kl}(g_{ij}^{\mathbf{r}(0)})\mathbf{u} = \frac{1}{2}\sum_{\mathbf{r} \in \mathcal{N}}\sum_{i=1}^{N}\sum_{j=1}^{N}\left[r^{k}L^{k}r^{l}L^{l}+r^{l}L^{l}x_{j}^{k}+r^{k}L^{k}x_{j}^{l}\right]g_{ij}^{\mathbf{r}(0)}; \label{eq:identities_2:c}
\end{equation}
\begin{equation}
\mathbf{u}^{T}B_{0}^{k}(g_{ij}^{\mathbf{r}(0)})\mathbf{w}^{l} = \sum_{\mathbf{r} \in \mathcal{N}}\sum_{i=1}^{N}\sum_{j=1}^{N}r^{k}L^{k}w_{j}^{l}g_{ij}^{\mathbf{r}(0)}; \label{eq:identities_2:d}
\end{equation}
\begin{equation}
\mathbf{u}^{T}B_{0}^{k}(g_{ij}^{\mathbf{r}(0)})\frac{\partial\mathbf{w}^{l}}{\partial x^{k}}=\sum_{\mathbf{r} \in \mathcal{N}}\sum_{i=1}^{N}\sum_{j=1}^{N}\sum_{k=1}^{3}r^{k}L^{k}\frac{\partial w_{j}^{l}}{\partial x^{k}}g_{ij}^{\mathbf{r}(0)}. \label{eq:identities_2:e}
\end{equation}
\end{subequations}

Substituting from the identities in \cref{eq:identities_2:c,eq:identities_2:d} into \cref{eq:K_tensor} we have
\begin{equation}
\mathcal{K}^{kl}=\frac{1}{2}\sum_{\mathbf{r} \in \mathcal{N}}\sum_{i=1}^{N}\sum_{j=1}^{N}\left[r^{k}L^{k}r^{l}L^{l}+r^{l}L^{l}x_{j}^{k}+r^{k}L^{k}x_{j}^{l}+r^{k}L^{k}w_{j}^{l}+r^{l}L^{l}w_{j}^{k}\right]g_{ij}^{\mathbf{r}(0)}\label{eq:K_w}
\end{equation}

Following \cite{chapman2021}, we choose 
\begin{equation}
\mathbf{w}^{l}=-L^{l}\mathbf{v}^{l}-\mathbf{X}^{l},
\label{eq:v_def}
\end{equation}
where
$\mathbf{X}^{l}=(x_{1}^{l},...,x_{N}^{l})^{T}$ such that, from \cref{eq:w_sol}, we have
\begin{equation}
L^{l}A_{0}(g_{ij}^{\mathbf{r}(0)})\mathbf{v}^{l}=-A_{0}(g_{ij}^{\mathbf{r}(0)})\mathbf{X}^{l}+B_{0}^{l}(g_{ij}^{\mathbf{r}(0)})\mathbf{u}=\sum_{\mathbf{r} \in \mathcal{N}}\sum_{j=1}^{N}r^{l}L^{l}g_{ij}^{\mathbf{r}(0)}.
\label{eq:v_vector}
\end{equation}
Substituting from \cref{eq:v_def} into \cref{eq:K_w} we have
\begin{equation}
\mathcal{K}^{kl}=\frac{L^{k}L^{l}}{2L^{1}L^{2}L^{3}}\sum_{\mathbf{r} \in \mathcal{N}}\sum_{i=1}^{N}\sum_{j=1}^{N}\left[r^{k}r^{l}+r^{k}v_{i}^{l}+r^{l}v_{i}^{k}\right]g_{ij}^{\mathbf{r}(0)},
\end{equation}
where we have relabelled $i\leftrightarrow j$ and $\mathbf{r}\leftrightarrow-\mathbf{r}$
and then used the identity in \cref{eq:g_period}. We have also divided \cref{eq:solve_p_O_2}
by $L^{1}L^{2}L^{3}$ in order that $\mathcal{K}^{kl}$ will represent the permeability
tensor [see \cref{eq:U_macro_0} \textit{et seq.}].

It has been shown in \cite{chapman2021} that
\begin{equation}
\sum_{\mathbf{r} \in \mathcal{N}}\sum_{i=1}^{N}\sum_{j=1}^{N}r^{k}v_{i}^{l}g_{ij}^{\mathbf{r}(0)}=\sum_{\mathbf{r} \in \mathcal{N}}\sum_{i=1}^{N}\sum_{j=1}^{N}r^{l}v_{i}^{k}g_{ij}^{\mathbf{r}(0)}.
\end{equation}
Therefore, we can also write $\mathcal{K}^{kl}$ as
\begin{equation}
\mathcal{K}^{kl}=\frac{L^{k}L^{l}}{2L^{1}L^{2}L^{3}}\sum_{\mathbf{r} \in \mathcal{N}}\sum_{i=1}^{N}\sum_{j=1}^{N}\left[r^{k}r^{l}+2r^{k}v_{i}^{l}\right]g_{ij}^{\mathbf{r}(0)}.
\label{eq:K_kl}
\end{equation}

Substituting from \cref{eq:identities_2:a,eq:identities_2:b,eq:identities_2:e} into \cref{eq:xi_vector} we have
\begin{equation}
\xi^{l}=\frac{L^{l}}{L^{1}L^{2}L^{3}}\sum_{\mathbf{r} \in \mathcal{N}}\sum_{i=1}^{N}\sum_{j=1}^{N}\sum_{s=1}^{3}r^{s}L^{s}\left[\left(\frac{1}{2}r^{l}+v_{i}^{l}\right)\frac{\partial g_{ij}^{\mathbf{r}(0)}}{\partial x^{s}}+\frac{\partial v_{i}^{l}}{\partial x^{s}}g_{ij}^{\mathbf{r}(0)}\right]
\end{equation}
Taking the divergence of \cref{eq:K_kl} we have
\begin{equation}
\begin{aligned}
\frac{\partial \mathcal{K}^{kl}}{\partial x^{k}} & =\frac{L^{l}}{L^{1}L^{2}L^{3}}\sum_{\mathbf{r} \in \mathcal{N}}\sum_{i=1}^{N}\sum_{j=1}^{N}\sum_{k=1}^{3}r^{k}L^{k}\left[\left(\frac{1}{2}r^{l}+v_{i}^{l}\right)\frac{\partial g_{ij}^{\mathbf{r}(0)}}{\partial x^{k}}+\frac{\partial v_{i}^{l}}{\partial x^{k}}g_{ij}^{\mathbf{r}(0)}\right]\\
 & =\xi^{l}.
\end{aligned}
\end{equation}
Therefore, \cref{eq:Darcy_modified_sym} can be written as the
Darcy equation
\begin{equation}
\boldsymbol{\nabla} \cdot \left( \boldsymbol{\mathcal{K}} \cdot \boldsymbol{\nabla} p\right) = 0.
\label{eq:Darcy}
\end{equation}

\subsection{Macroscopic fluid velocity}

Following Chapman and Wilmott \cite{chapman2021}, the  
macroscopic Darcy velocity
in the $k$ direction, $U^{k}$, is given by
\begin{equation}
U^{k}=\frac{L^{k}}{2L^{1}L^{2}L^{3}}\sum_{\mathbf{r} \in \mathcal{N}}\sum_{i=1}^{N}\sum_{j=1}^{N}r^{k}q_{ij}^{\mathbf{r}}.
\label{eq:U_macro_0}
\end{equation}

Using the scaling introduced in \cref{sec:blood_flow_model}, equation \cref{eq:U_macro_0} defines the scale of the Darcy velocity as
\[
\mathbf{U^*} =  \frac{\pi}{128}\frac{d_0^{*4} \Delta p^*_0}{l_0^{*2} \Lambda^*_0 \mu^*_0} \mathbf {U},
\]
which is equivalent to the characteristic scale of the flow in a vessel divided by the scale of the cross-sectional area of the unit cell.
Substituting from \cref{eq:q_p} into \cref{eq:U_macro_0}, expanding $\mathbf{P}$
and $g^{\bf r}_{ij}$ as in \cref{eq:p_exp,eq:g_ij_exp}, and substituting from \cref{eq:p_0,eq:p_1}, we obtain
\begin{equation}
\begin{aligned}U^{k}= & -\frac{L^{k}}{2L^{1}L^{2}L^{3}\epsilon}\sum_{\mathbf{r} \in \mathcal{N}}\sum_{i=1}^{N}\sum_{j=1}^{N}r^{k}\left(g_{ij}^{\mathbf{r}}+\frac{1}{2}\epsilon\sum_{s=1}^{3}(x_{i}^{s}+x_{j}^{s}+r^{s}L^{s})\frac{\partial g_{ij}^{\mathbf{r}}}{\partial x^{s}} + O(\epsilon^2) \right)\\
&\qquad  \qquad \qquad \qquad \times\Bigg[p_{j}+\epsilon\sum_{l=1}^{3}(x_{j}^{l}+r^{l}L^{l})
\frac{\partial p_{j}}{\partial x^{l}}-p_{i}-\epsilon\sum_{l=1}^{3}x_{i}^{l}\frac{\partial p_{i}}{\partial x^{l}}+ O(\epsilon^2) \Bigg]\\
= & -\frac{L^{k}}{2L^{1}L^{2}L^{3}}\sum_{\mathbf{r} \in \mathcal{N}}\sum_{i=1}^{N}\sum_{j=1}^{N}r^{k}g_{ij}^{\mathbf{r}(0)} \\
& \qquad  \qquad \qquad \qquad \times \Biggl[p_{j}^{(1)}+\sum_{l=1}^{3}(x_{j}^{l}+r^{l}L^{l})\frac{\partial p_{j}^{(0)}}{\partial x^{l}}-p_{i}^{(1)}-\sum_{l=1}^{3}x_{i}^{l}\frac{\partial p_{i}^{(0)}}{\partial x^{l}}\Biggr]+O(\epsilon)\\
= & -\frac{L^{k}}{2L^{1}L^{2}L^{3}}\sum_{\mathbf{r} \in \mathcal{N}}\sum_{i=1}^{N}\sum_{j=1}^{N}\sum_{l=1}^{3}r^{k}g_{ij}^{\mathbf{r}(0)}\left[w_{j}^{l}+x_{j}^{l}+r^{l}L^{l}-w_{i}^{l}-x_{i}^{l}\right]\frac{\partial p}{\partial x^{l}}+O(\epsilon)\\
= & -\frac{L^{k}L^{l}}{2L^{1}L^{2}L^{3}}\sum_{\mathbf{r} \in \mathcal{N}}\sum_{i=1}^{N}\sum_{j=1}^{N}\sum_{l=1}^{3}r^{k}g_{ij}^{\mathbf{r}(0)}\left[v_{i}^{l}-v_{j}^{l}+r^{l}\right]\frac{\partial p}{\partial x^{l}}+O(\epsilon)\\
= & -\frac{L^{k}L^{l}}{2L^{1}L^{2}L^{3}}\sum_{\mathbf{r} \in \mathcal{N}}\sum_{i=1}^{N}\sum_{j=1}^{N}\sum_{l=1}^{3}g_{ij}^{\mathbf{r}(0)}\left[2r^{k}v_{i}^{l}+r^{k}r^{l}\right]\frac{\partial p}{\partial x^{l}}+O(\epsilon).
\end{aligned}
\label{eq:exp_U}
\end{equation}
In the fourth line we substituted $\mathbf{w}^l$ from \cref{eq:v_def} and in the fifth line we relabeled $j\rightarrow i$ and $\mathbf{r}\leftrightarrow-\mathbf{r}$
in the $v_j^l$ term, and then used the identity in \cref{eq:g_period}.
Combining \cref{eq:exp_U,eq:K_kl}, the Darcy velocity can be written as
\begin{equation}
\mathbf{U}=-\boldsymbol{\mathcal{K}} \cdot \boldsymbol{\nabla} p.
\label{eq:U_macro_K}
\end{equation}
Transforming back to dimensional variables, we have that the dimensional Darcy velocity is given by
\begin{equation}
\mathbf{U^*}=-\frac{\mathcal{K}_0^*}{\mu_0^*}\boldsymbol{\mathcal{K}} \cdot \boldsymbol{\nabla^*} p^*,
\label{eq:U_macro_dimensional}
\end{equation}
where
\[
\mathcal{K}_0^*=\frac{\pi}{128}\frac{d_0^{*4}}{l_0^{*2}}.
\]
Therefore, we verify that the macroscopic permeability has the appropriate dimensions of length squared.

We see from \cref{eq:Darcy,eq:U_macro_K} that the blood pressure and velocity at the \linebreak macroscale can be solved if the macroscopic permeability tensor, $\boldsymbol{\mathcal{K}}$, is known. \linebreak According to \cref{eq:K_kl}, $\boldsymbol{\mathcal{K}}$ depends on the leading-order vessel conductance
\begin{equation}
g_{ij}^{\mathbf{r}(0)}=\frac{d_{ij}^{4}}{l_{ij}\mu(h_{ij}^{\mathbf{r}(0)},\widetilde{d}_{ij})},
\label{eq:leading_order_conductance}
\end{equation}
which is a function of the leading order of the vessel haematocrit, $h_{ij}^{\mathbf{r}(0)}$. Consequently, we now turn to evaluate $h_{ij}^{\mathbf{r}(0)}$.

\section{Homogenisation of the haematocrit propagation} \label{sec:homog_haematocrit}

Under steady flow conditions, the haematocrit is uniform along a vessel [see \cref{eq:h_conserv}], although it changes between vessels due to the nonlinear
effect of phase separation occurring at flow bifurcations. In what
follows, we consider three different type of vessels that may exist within each
unit cell (see illustration in \cref{fig:vessel_types}): (i) vessels that originate from bifurcating flow junctions;
(ii) vessels that originate from converging flow junctions; (iii) vessels that originate from nodes in neighbouring unit cells and terminate inside the unit cell (henceforth we will refer to these vessels as ``inlet vessels''). Our goal is to determine a relation of the type in \cref{eq:h_general} between the haematocrit in a given vessel and the haematocrit and flow in its parent vessel(s). In what follows we explicitly evaluate the functional form of this relation for the three aforementioned types of vessels.

\subsection{Bifurcating flow junctions}

It was observed experimentally (e.g., \cite{carr1991influence,fenton1985nonuniform,klitzman1982capillary,pries1989red}) that, at flow bifurcations, the RBC component of blood does not generally split in proportion to the flow splitting; instead, RBC splitting appears to be a nonlinear function of the ratio of the flow rates at the junction, the local geometry of the vessels (specifically the diameters), and the haematocrit entering the bifurcating junction. The model of Pries et al. \cite{pries1989red} is a popular nonlinear empirical relation for representing haematocrit
splitting at flow bifurcations. We will use the \linebreak
aforementioned model with the later refinement of coefficients that was carried out in \cite{pries2005microvascular}. Henceforth, we will denote this model as the Pries model for brevity. Using the Pries model, we have that the ratio of the haematocrit flux, $h_J q_J$, entering a vessel $J$ to the haematocrit flux in its parent vessel, $h_{J_p} q_{J_p}$, satisfies
\begin{equation}
\frac{h_{J}}{h_{J_{p}}}\psi_{J/J_{p}}=\frac{\exp\left[ a_{J_{p}}D\left(d_{J}/d_{J'}\right)\right](\psi_{J/J_{p}}- c_{J_{p}})^{1+b_{J_{p}}}}{\exp\left[ a_{J_{p}}D\left(d_{J}/d_{J'}\right)\right](\psi_{J/J_{p}}- c_{J_{p}})^{1+ b_{J_{p}}}+(1-\psi_{J/J_{p}}- c_{J_{p}})^{1+ b_{J_{p}}}}.
\label{eq:f_first}
\end{equation}
Here, $\psi_{J/J_{p}}=q_{J}/q_{J_{p}}\in[0,1]$, where $q_{J}$
and $q_{J_{p}}$ represent the (absolute) magnitudes of the flows in vessels $J$ and $J_p$, respectively. The subscript $J'$ represents the sister branch of vessel $J$, while the known function $D(\cdot)$ has the property that $D\left(d_{J}/d_{J'}\right)=-D\left(d_{J'}/d_{J}\right)$ \cite{pries1989red,pries2005microvascular}. The haematocrit and flow in \cref{eq:f_first} are evaluated at specific vessel locations (either, $\mathbf{x}_J$ or $\mathbf{x}_{J_p}$). Additionally, $a_{J_{p}}$, $b_{J_{p}}$ and
$c_{J_{p}}$ are functions of the haematocrit and diameter of the
parent vessel
\[
\frac{a_{J_{p}}}{a}=\frac{b_{J_{p}}}{b}=\frac{c_{J_{p}}}{c}=\frac{1-h_{J_{p}}}{\widetilde{d}_{J_{p}}},
\]
where $a$, $b$, and $c$ are $O(1)$ constants and we use the micron-unit-scaled
diameter, $\widetilde{d}_{J_{p}}$, for consistency with the functional form of the Pries model (similar to the viscosity model by \cite{pries1994resistance}; see \cref{Poiseuille}). Most vessels have diameters much larger than one micron, so that $\widetilde{d}_{J_{p}} \gg 1$. Therefore, we may introduce the small parameter, $\alpha \epsilon$, to the definition of  $\widetilde{d}_{J_{p}}$as follows 
\begin{equation}
\widetilde{d}_{J_{p}}=\frac{d_{J_{p}}^{*}d^*_{0}}{d_{\mu}^*d^*_{0}}=\frac{1}{\alpha\epsilon}d_{J_{p}},
\label{eq:rescaling_d}
\end{equation}
where $d^*_{0}$ is the reference diameter of the vessels in the network ($\gg 1\unit{\mu m}$) and the  nondimensional, $O(1)$, vessel diameter is
\[
d_{J_{p}}=\frac{d_{J_{p}}^{*}}{d^*_{0}}.
\]
The small parameter is defined by
\[
\alpha\epsilon=\frac{d_{\mu}^*}{d^*_{0}}\ll1.
\]
Therefore, using the definition of $\epsilon$ in \cref{eq:epsilon}, we have that
\[
\alpha = \frac{(d_{\mu}^*/d^*_{0})}{(l_0^*/\Lambda_0^*)},
\]
representing the ratio of the two small parameters in the asymptotic expansion.

Introducing the rescaling in \cref{eq:rescaling_d} into \cref{eq:f_first} we have
\begin{equation}
\begin{aligned}
\frac{h_{J}}{h_{J_{p}}}\psi_{J/J_{p}}= & \exp\left[\alpha\epsilon a_{J_{p}}D\left(d_{J}/d_{J'}\right)\right](\psi_{J/J_{p}}-\alpha\epsilon c_{J_{p}})^{1+\alpha\epsilon b_{J_{p}}} \\
& \times \Bigg[ \exp\left[\alpha\epsilon a_{J_{p}}D\left(d_{J}/d_{J'}\right)\right](\psi_{J/J_{p}}-\alpha\epsilon c_{J_{p}})^{1+\alpha\epsilon b_{J_{p}}} \\
& \qquad  +(1-\psi_{J/J_{p}}-\alpha\epsilon c_{J_{p}})^{1+\alpha\epsilon b_{J_{p}}} \Bigg]^{-1}.
\end{aligned}
\label{eq:f_rescaled}
\end{equation}
Expanding \cref{eq:f_rescaled} in powers of $\epsilon$ and dividing by $\psi_{J/J_{p}}$ we have
\begin{equation}
\begin{aligned}
\frac{h_{J}}{h_{J_{p}}}=&1+\alpha\epsilon\left(\frac{1-h_{J_{p}}}{d_{J_{p}}}\right)\Bigg[b(1-\psi_{J/J_{p}})\ln\left(\frac{\psi_{J/J_{p}}}{1-\psi_{J/J_{p}}}\right)\\
&+c\left(2-\frac{1}{\psi_{J/J_{p}}}\right) + a (1-\psi_{J/J_{p}})D\left(d_{J}/d_{J'}\right)\Bigg]+O\left(\epsilon^2\right).
\end{aligned}
\label{eq:linearized_splitting}
\end{equation}
We note that this expansion holds provided that $D\left(d_{J}/d_{J'}\right) \lesssim O(1)$ and that $\psi_{J/J_{p}}$ is
not in the neighborhood of $\alpha\epsilon c_{J_{p}}$ or $1-\alpha\epsilon c_{J_{p}}$, namely $|\ln\left(\psi_{J/J_{p}}/(1-\psi_{J/J_{p}})\right)| \lesssim O(1)$. 

From \cref{eq:linearized_splitting} we note that, at leading order, the haematocrit in the daughter branches
equals the haematocrit in the parent branch. Consequently, the leading-order haematocrit is uniform across the unit cell. Accordingly, we now expand the haematocrit around the reference state in which the haematocrit is equal to $h(\mathbf{x})$ for all vessels in the unit cell
\begin{equation}
h_{J}=h+\epsilon h_{J}^{(1)}+O(\epsilon^{2}).
\label{eq:expansion_h_q}
\end{equation}
Similarly, we expand the flow around its leading order value
\begin{equation}
q_{J}=q_{J}^{(0)} + \epsilon q_{J}^{(1)}+O(\epsilon^{2}).
\label{eq:expansion_q}
\end{equation}
The absolute magnitude of the flow associated with the reference state, $q_{J}^{(0)}$, can be derived from the leading
order flow problem. Expanding \cref{eq:q_p} in a manner similar to that used to expand \cref{eq:exp_U} we have
\begin{equation}
q_{J}^{(0)} \equiv |q_{j_{0}j_{1}}^{\mathbf{r}(0)}|=g_{J}(h)\left|\sum_{l=1}^{3} L^{l} \left[v_{j_{0}}^{l}-v_{j_{1}}^{l}+r_{J}^{l}\right]\frac{\partial p}{\partial x^{l}} \right|,
\label{eq:q_J}
\end{equation}
where $j_{0}$ and $j_{1}$ are the end nodes of vessel $J$, and $\mathbf{r}_{J}$ is the cell-neighbour indicator associated with
vessel $J$. In \cref{eq:q_J}, we have used the fact that the leading order vessel \linebreak
conductance is a function of the leading order haematocrit, $g_{j_0j_1}^{\mathbf{r}(0)} \equiv g_J(h)$.

We note here that if $j_{0}$ and $j_{1}$ are the edge nodes of vessel
$J$, we define $\mathbf{x}_{J}$ as
\begin{equation}
\mathbf{x}_{J}=\frac{1}{2}(\mathbf{x}_{j_{0}}+\mathbf{x}_{j_{1}}+\mathbf{r}_{J}\circ\mathbf{L}) \equiv \overline{\mathbf{x}}^{\mathbf{r}_J}_{j_0 j_1},
\end{equation}
in accordance with the definition in \cref{eq:x_ij}.
We now expand the haematocrit in vessel $J$ located at $\mathbf{x}_{J}$
around the cell origin, $\mathbf{x}$:
\begin{equation}
h_{J}(\mathbf{x}+\epsilon\mathbf{x}_{J})=h_{J}(\mathbf{x})+\epsilon\sum_{s=1}^{3}x_{J}^{s}\frac{\partial h_{J}}{\partial x^{s}}|_{\mathbf{x}}+O(\epsilon^{2}).
\label{eq:x_exp_h}
\end{equation}

The $O(\epsilon)$ haematocrit in the daughter branches is then determined
by substituting \cref{eq:expansion_h_q,eq:expansion_q,eq:q_J,eq:x_exp_h} into \cref{eq:linearized_splitting}
\begin{equation}
\begin{aligned} 
 q_J^{(0)}\left( h_{J_{p}}^{(1)} - h_{J}^{(1)} \right) & =
   q_J^{(0)}\sum_{s=1}^{3}(x_{J}^{s}-x_{J_{p}}^{s})\frac{\partial h}{\partial x^{s}}\\
 & - \alpha h\left(\frac{1-h}{d_{J_{p}}}\right)q_J^{(0)}\Biggl[ b(1-\psi_{J/J_{p}}^{(0)})\ln\left(\frac{\psi_{J/J_{p}}^{(0)}}{1-\psi_{J/J_{p}}^{(0)}}\right) \\
 & +c\left(2-\frac{1}{\psi_{J/J_{p}}^{(0)}}\right)+a (1-\psi_{J/J_{p}}^{(0)}) D\left(d_{J}/d_{J'}\right)  
 \Biggr].
\end{aligned}
\label{eq:h_1}
\end{equation}
For brevity we write \cref{eq:h_1} as
\begin{equation}
q_J^{(0)}\left( h_{J_{p}}^{(1)} - h_{J}^{(1)} \right)=q_J^{(0)}\sum_{s=1}^{3}(x_{J}^{s}-x_{J_p}^{s})\frac{\partial h}{\partial x^{s}}-\alpha h(1-h)f_{J}^{(0)},
\label{eq:bifuracting_vessels}
\end{equation}
where
\begin{equation}
\begin{aligned} 
f_J^{(0)} & = \frac{q_J^{(0)}}{d_{J_{p}}}\Biggl[ b(1-\psi_{J/J_{p}}^{(0)})\ln\left(\frac{\psi_{J/J_{p}}^{(0)}}{1-\psi_{J/J_{p}}^{(0)}}\right) \\
 & +c\left(2-\frac{1}{\psi_{J/J_{p}}^{(0)}}\right)+a (1-\psi_{J/J_{p}}^{(0)}) D\left(d_{J}/d_{J'}\right)  
 \Biggr]
 \end{aligned}
 \label{eq:f_J_first}
\end{equation}
is a function containing details of the splitting rule at leading order.

\begin{lemma}
 $f_J^{(0)}=-f_{J'}^{(0)}$
\label{lemma_1}
 \end{lemma}

 \begin{proof}

 Starting from equation \cref{eq:f_J_first} and replacing $\psi_{J/J_{p}}^{(0)}=1-\psi_{J'/J_{p}}^{(0)}$ in the first and third terms of the equation, and using $\psi_{J/J_{p}}^{(0)}=q_J^{(0)}/q_{J_{p}}^{(0)}$ in the second term, we have
\begin{equation}
\begin{aligned} 
f_J^{(0)} & = \frac{1}{d_{J_{p}}}\Biggl[  -b q_J^{(0)}\psi_{J'/J_{p}}^{(0)}\ln\left(\frac{\psi_{J'/J_{p}}^{(0)}}{1-\psi_{J'/J_{p}}^{(0)}}\right) \\
 & +c\left(2q_J^{(0)}-q_{J_{p}}^{(0)}\right) + a  D\left(d_{J}/d_{J'}\right) q_J^{(0)}\psi_{J'/J_{p}}^{(0)}  
 \Biggr].
 \end{aligned}
\end{equation}

Using $q_J^{(0)} = q_{J_{p}}^{(0)}-q_{J'}^{(0)}$ we have
\begin{equation}
\begin{aligned} 
f_J^{(0)} & = \frac{1}{d_{J_{p}}}\Biggl[  -b  (q_{J_{p}}^{(0)}-q_{J'}^{(0)})\psi_{J'/J_{p}}^{(0)}\ln\left(\frac{\psi_{J'/J_{p}}^{(0)}}{1-\psi_{J'/J_{p}}^{(0)}}\right) \\
 & - c\left(2q_{J'}^{(0)}-q_{J_{p}}^{(0)}\right) + a  D\left(d_{J}/d_{J'}\right) q_{J'}^{(0)}(1-\psi_{J'/J_{p}}^{(0)})  
 \Biggr].
 \end{aligned}
\end{equation}
Using $\psi_{J'/J_{p}}^{(0)}=q_{J'}^{(0)}/q_{J_{p}}^{(0)}$ in the first and second terms and $ D\left(d_{J}/d_{J'}\right)=- D\left(d_{J'}/d_{J}\right)$ \cite{pries2005microvascular} in the third term we have
\begin{equation}
\begin{aligned} 
f_J^{(0)} & = -\frac{ q_{J'}^{(0)}}{d_{J_{p}}}\Biggl[  b(1-\psi_{J'/J_{p}}^{(0)})\ln\left(\frac{\psi_{J'/J_{p}}^{(0)}}{1-\psi_{J'/J_{p}}^{(0)}}\right) \\
 & + c \left(2-\frac{1}{\psi_{J'/J_{p}}^{(0)}}\right) + a D\left(d_{J'}/d_{J}\right)(1-\psi_{J'/J_{p}}^{(0)})  
 \Biggr].
 \end{aligned}
 \label{eq:f_J_final}
\end{equation}
On the other hand, replacing $J \rightarrow J'$ in \cref{eq:f_J_first} we have
\begin{equation}
\begin{aligned} 
f_{J'}^{(0)} & = \frac{ q_{J'}^{(0)}}{d_{J_{p}}}\Biggl[  b(1-\psi_{J'/J_{p}}^{(0)})\ln\left(\frac{\psi_{J'/J_{p}}^{(0)}}{1-\psi_{J'/J_{p}}^{(0)}}\right) \\
 & + c \left(2-\frac{1}{\psi_{J'/J_{p}}^{(0)}}\right) + a (1-\psi_{J'/J_{p}}^{(0)}) D\left(d_{J'}/d_{J}\right)  
 \Biggr].
 \end{aligned}
  \label{eq:f_J_prime}
\end{equation}

Comparing \cref{eq:f_J_final,eq:f_J_prime} we have that $f_J^{(0)}=-f_{J'}^{(0)}$ as required.
\end{proof}

\subsection{Converging flow junctions}

At nodes where the flow in two vessels \linebreak
converge into one, we have the following haematocrit flow
balance
\begin{equation}
q_{J}h_{J}|_{\mathbf{x}+\epsilon\mathbf{x}_{J}} = \sum_{J_{p}}q_{J_{p}}h_{J_{p}}|_{\mathbf{x}+\epsilon\mathbf{x}_{J_p}},
\end{equation}
where the summation is over the two parent vessels of vessel $J$. We expand $h_J$ and $q_J$ around their
reference state and expand $h_J(\mathbf{x}+\epsilon\mathbf{x}_J)$ around $\mathbf{x}$, as for \cref{eq:expansion_h_q,eq:expansion_q,eq:q_J,eq:x_exp_h}. We note that in \cref{eq:q_J}, $q_J^{(0)}$ is already given in $\mathbf{x}+\epsilon\mathbf{x}_J$ as a function of the leading order pressure and haematocrit, hence, it is not necessary to expand it around the cell origin. Accordingly, we obtain
\begin{equation}
\sum_{J_{p}}q_{J_p}^{(0)}h_{J_{p}}^{(1)}-q_{J}^{(0)}h_{J}^{(1)} = \sum_{s=1}^{3} \left(x_{J}^{s} q_{J}^{(0)} - \sum_{J_{p}}x_{J_{p}}^{s}q_{J_p}^{(0)} \right)\frac{\partial h}{\partial x^{s}},
\label{eq:convering_vessels}
\end{equation}
where we have used the fact that the net flow balance at each node is zero, i.e., $\sum_{J_{p}}q_{J_{p}}-q_{J}=0$.

\subsection{Inlet vessels}

Recalling the \hyperref[sec:sidebar]{sidebar} discussion in \cref{sec:homog_flow} regarding the periodicity of the network, we assume
\begin{equation}
h_{ij}^{\bf r}(\mathbf{x})=h_{ji}^{\bf -r}(\mathbf{x}),
\label{eq:periodicity_h_ij}
\end{equation}
in accordance with \cref{eq:g_period}. Transforming from node to vessel notations, \cref{eq:periodicity_h_ij} states that if $J$ is an inlet vessel having corresponding periodic outlet vessel
$\hat{J}$, then we have
\begin{equation}
h_{J}(\mathbf{x})=h_{\hat{J}}(\mathbf{x}),
\label{eq:periodicity_h}
\end{equation}
Taking the $O(\epsilon)$ terms from \cref{eq:periodicity_h} and multiplying by $q_{J}^{(0)}$ we have
\begin{equation}
q_{J}^{(0)}\left( h_{\hat{J}}^{(1)} - h_{J}^{(1)} \right) = 0.
\label{eq:inlet_vessels}
\end{equation}

\subsection{Asymptotic solution for the haematocrit}

Combining all types of \linebreak vessels, \cref{eq:bifuracting_vessels,eq:convering_vessels,eq:inlet_vessels}, into a single relation we have
\begin{equation}
\begin{aligned} 
& Y_{J}^{\mathrm{sp}}q_J^{(0)}\left( h_{J_{p}}^{(1)} - h_{J}^{(1)} \right) +  Y_{J}^{\mathrm{co}}\left(\sum_{J_{p}}q_{J_p}^{(0)}h_{J_{p}}^{(1)}-q_{J}^{(0)}h_{J}^{(1)} \right) + Y_{J}^{\mathrm{in}}q_{J}^{(0)}\left( h_{\hat{J}}^{(1)} - h_{J}^{(1)} \right)=\\
& Y_{J}^{\mathrm{sp}}\left(q_J^{(0)}\sum_{s=1}^{3}(x_{J}^{s}-x_{J_p}^{s})\frac{\partial h}{\partial x^{s}}-\alpha h(1-h)f_{J}^{(0)}\right) \\
& + Y_{J}^{\mathrm{co}}\sum{}_{s=1}^{3} \left(x_{J}^{s} q_{J}^{(0)} - \sum_{J_{p}}x_{J_{p}}^{s}q_{J_p}^{(0)} \right)\frac{\partial h}{\partial x^{s}},\;\;\;
\text{for }J=1,...,M.
\end{aligned}
\label{eq:h_splitting_merging}
\end{equation}
Here, $Y_{J}^{\mathrm{sp}}=1$ if vessel $J$ originates from a splitting
bifurcation and $Y_{J}^{\mathrm{sp}}=0$ otherwise; $Y_{J}^{\mathrm{co}}=1$ if vessel
$J$ originates from a converging bifurcation and  $Y_{J}^{\mathrm{co}}=0$ otherwise;
$Y_{J}^{\mathrm{in}}=1$ if vessel $J$ is an inlet vessel and $Y_{J}^{\mathrm{in}}=0$
otherwise. If we write the first-order vessel haematocrit as a vector, $\mathbf{h}^{(1)}=(h_{1}^{(1)},...,h_{M}^{(1)})^{T}$, then we obtain a system of $M$ equations for $h_J^{(1)}$ ($J=1,...,M$), which can be written in a matrix form as follows:
\begin{equation}
\mathbf{Q}\mathbf{h}^{(1)}= \sum_{s=1}^{3}\boldsymbol{\eta}^{s}\frac{\partial h}{\partial x^{s}} - \alpha h(1-h)\boldsymbol{\varpi},
\label{eq:h_1_vector_eq}
\end{equation}
where the matrix $\mathbf{Q}$ is given by
\begin{equation}
Q_{JK}=\begin{cases}
-q_J^{(0)}, & \  \text{if \ensuremath{J=K}}\\
\;\;\; q_J^{(0)}, & \ \text{if \ensuremath{K} is the parent vessel of \ensuremath{J} at a splitting junction} \\
& \ \  \text{or \ensuremath{K} is the periodic outlet vessel of inlet \ensuremath{J}} \\
\;\;\; q_{K}^{(0)}, &\ \text{if} \ K \text{ is the parent vessel of \ensuremath{J} at a converging junction}
\end{cases}
\end{equation}
The vectors $\boldsymbol{\eta}^{s}$ and $\boldsymbol{\varpi}$ read
\begin{equation}
\begin{aligned}\eta_{J}^{s}= & Y_{J}^{\mathrm{sp}}q_J^{(0)}(x_{J}^{s}-x_{J_p}^{s}) + Y_{J}^{\mathrm{co}}\left(q^{(0)}_{J}x_{J}^{s} - \sum_{J_{p}}q^{(0)}_{J_p}x_{J_{p}}^{s}\right) \;\; \text{and} \\
\varpi_{J}= & Y_{J}^{\mathrm{sp}}f_{J}^{(0)}.
\end{aligned}
\label{eq:eta_zeta_etc}
\end{equation}

It is straightforward to show that each column $K$ of $Q_{JK}$ has either 2 or 3 nonzero entries according to how many daughter vessels it has: (i) $(-q_K^{(0)}, q_K^{(0)})$ if $K$ is either an outlet vessel or one of the parent vessels in a converging-flow junction; (ii) $(-q_K^{(0)}, q_{K_{d_{1}}}^{(0)}, q_{K_{d_{2}}}^{(0)})$, if $K$ is a parent vessel at a splitting junction where $K_{d_{1}}$ and $K_{d_{2}}$ are the indices of the daughter vessels. Therefore, due to conservation of flow at bifurcations, we have
\begin{equation}
\mathbf{Q}^T\mathbf{u} = \boldsymbol{0},
\end{equation}
where we have used the $[M \times 1]$ vector $\mathbf{u}=(1,...,1)^{T}$.
Accordingly, we can use the Fredholm alternative to write the solvability condition of \cref{eq:h_1_vector_eq} as a partial differential equation for the propagation of the leading-order haematocrit at the macroscale,
\begin{equation}
\mathbf{u}^{T}\left(\sum_{s=1}^{3}\boldsymbol{\eta}^{s}\frac{\partial h}{\partial x^{s}}  - \alpha h(1-h)\boldsymbol{\varpi} \right)=0.\label{eq:solv_cond}
\end{equation}
Finally, \cref{eq:solv_cond} can be simplified to obtain
\begin{equation}
\boldsymbol{\mathcal{C}}\cdot\boldsymbol{\nabla}h=\alpha h(1-h)\mathcal{S} ,
\label{eq:h_macro_full}
\end{equation}
where
\begin{equation}
\mathcal{C}^s=\mathbf{u}^T\boldsymbol{\eta}^{s},\;\;\text{and}\;\;\mathcal{S}=\mathbf{u}^T\boldsymbol{\varpi}.
\label{eq:C_and_S}
\end{equation}

\begin{lemma}
$\mathcal{S}=0$
\label{lemma_2}
\end{lemma}

\begin{proof}
Combining \cref{eq:eta_zeta_etc,eq:C_and_S} we have
\begin{equation}
\mathcal{S} =\mathbf{u}^T\boldsymbol{\varpi}=\sum_{J=1}^{M}Y_{J}^{\mathrm{sp}}f_{J}^{(0)}.
\end{equation}

From \cref{lemma_1} we have that daughter vessels have equal and opposite values of $f_J^{(0)}$, such that they cancel each other when all splitting vessels are summed together. Consequently we have $\mathcal{S}=0$.
\end{proof}

\newpage

\begin{lemma}
$\boldsymbol{\mathcal{C}}/(L^1L^2L^3)=\mathbf{U}$
\label{lemma_3}
\end{lemma}

\begin{proof}

Dividing $\mathcal{C}^s$ by $(L^1L^2L^3)$ and using \cref{eq:eta_zeta_etc} we have
\begin{equation}
\begin{aligned}
& \frac{\mathcal{C}^s}{L^1L^2L^3}=\frac{\mathbf{u}^T\boldsymbol{\eta}^{s}}{L^1L^2L^3} = \\
& \frac{1}{L^1L^2L^3}\sum_{J=1}^{M}\left[ Y_{J}^{\mathrm{sp}}q_{J}^{(0)}(x_{J}^{s} - x_{J_{p}}^{s})+Y_{J}^{\mathrm{co}}\left(q_{J}^{(0)}x_{J}^{s} - \sum_{J_{p}}q_{J_p}^{(0)}x_{J_{p}}^{s}\right) \right].
\end{aligned}
\label{eq:proof_U_1}
\end{equation}
We note that internal vessels do not contribute to this sum because they are both daughter and parent vessels such that their contributions, appearing with both \linebreak
positive and negative signs, are cancelled out. Therefore, \cref{eq:proof_U_1} can be written as
\begin{equation}
\frac{\mathcal{C}^s}{L^1L^2L^3} =  \frac{1}{L^1L^2L^3}\sum_{J=1}^{M}Y^{\mathrm{in}}q_{J}^{(0)}(x_{\hat{J}}^{s} - x_{J}^s)= \frac{L^s}{2L^1L^2L^3}\sum_{\mathbf{r}}\sum_{i=1}^{N}\sum_{j=1}^{N} r^sq_{ij}^{\mathbf{r}(0)} = U^s,
\end{equation}
where $\hat{J}$ is the periodic outlet vessel corresponding to inlet vessel $J$, and we have used the definition of $U^s$ from \cref{eq:U_macro_0}. We have also used that, at leading order, $q_{J}^{(0)} = q_{\hat{J}}^{(0)}$ [see \cref{eq:q_J}]. Thus, we have proved that the Darcy velocity, $\mathbf{U}$, is equal to the haematocrit advection vector.
\end{proof}

Finally, we apply \cref{lemma_2,lemma_3} to equation \cref{eq:h_macro_full} to obtain the final form of the homogenised equation for haematocrit propagation
\begin{equation}
\mathbf{U}\cdot\boldsymbol{\nabla}h=0.
\label{eq:h_macro_final}
\end{equation}

\section{Summary of the homogenised model}
The following is the final \linebreak
homogenised model for macroscopic blood flow and haematocrit propagation \linebreak
developed in this manuscript:
\begin{subequations}\label{eq:final_model}
\begin{equation}
\boldsymbol{\nabla} \cdot \left( \boldsymbol{\mathcal{K}}(h) \cdot \boldsymbol{\nabla} p\right) = 0, \label{eq:Darcy_final}
\end{equation}
\begin{equation}
\mathbf{U}=-\boldsymbol{\mathcal{K}}(h)\cdot \boldsymbol{\nabla}p
\label{eq:Darcy_velocity_final}
\end{equation}
and
\begin{equation}
\mathbf{U}\cdot\boldsymbol{\nabla}h=0.
\label{eq:haematocrit_final}
\end{equation}
\end{subequations}
Here, $p(\mathbf{x})$, $\mathbf{U}(\mathbf{x})$, and $h(\mathbf{x})$ are the pressure, flow velocity, and haematocrit, \linebreak
respectively, at the macroscale; $\boldsymbol{\mathcal{K}}(h)$ is the macroscopic permeability tensor given by
\begin{equation}
\mathcal{K}^{kl}=\frac{L^{k}L^{l}}{2L^{1}L^{2}L^{3}}\sum_{\mathbf{r} \in \mathcal{N}}\sum_{i=1}^{N}\sum_{j=1}^{N}\frac{ d_{ij}^{4}\left(r^{k}r^{l}+2r^{k}v_{i}^{l}\right)}{l_{ij}\mu(h,\widetilde{d}_{ij})},
\label{eq:K_final}
\end{equation}
where $\mathbf{v}^{l}$ is a solution to the linear system of equations:
\begin{equation}
A_{0}(g_{ij}^{\mathbf{r}(0)})\mathbf{v}^{l}=\sum_{\mathbf{r} \in \mathcal{N}}\sum_{j=1}^{N}r^{l}g_{ij}^{\mathbf{r}(0)}.
\label{eq:v_vector_final}
\end{equation}
Here,
\[
A_{0}(g_{ij}^{\mathbf{r}(0)})=\sum_{\mathbf{r} \in \mathcal{N}}\left[g_{ij}^{\mathbf{r}(0)}-\delta_{ij}\sum_{m=1}^{N}g_{im}^{\mathbf{r}(0)}\right]
\]
and
\[
g_{ij}^{\mathbf{r}(0)} = \frac{d_{ij}^4}{l_{ij}\mu(h,\widetilde{d}_{ij})}.
\]

\section{Probability distribution of the macroscopic permeability\label{sec:permeability}}
The \linebreak
macroscopic permeability tensor, given by \cref{eq:K_final}, is a function of the macroscopic haematocrit, $h$, the distribution of vessel diameters and lengths, and the topology of the unit cell. In this section we illustrate how these features affect the components of the permeability tensor. We consider two different network topologies: two-dimensional honeycomb networks (\cref{sec:honeycomb}) and three-dimensional networks having vessels oriented in random directions (\cref{sec:random_3D}). For each network topology, we assume that the distribution of vessel diameters is log-normal, with the following probability distribution function:
\begin{equation}
f(\widetilde{d}) = \frac{1}{\widetilde{d}\sigma_{LN}\sqrt{2\pi}}\exp\left[-\frac{\left(\ln \widetilde{d}-\eta_{LN}\right)^2}{2\sigma^2_{LN}}\right],
\label{eq:log_normal}
\end{equation}
where $\widetilde{d} = d^*/d^*_{\mu}$ and $d^*_{\mu}=1\,\unit{\mu m}$.
For all results presented in this section, we fix $\sigma_{LN}=0.3$, and $\langle\widetilde{d}\rangle=20$ (here, and henceforth, angular brackets denote mean values), such that $\ln (\eta_{LN})=\langle\widetilde{d}\rangle e^{-\sigma^2_{LN}/2} \approx 19.12$. Using these parameters, the probability of obtaining $\widetilde{d}<10$ is very small ($\approx 1.5\%$), and most vessels satisfy $\widetilde{d} \gg 1$, as required for the validity of the asymptotic expansion [see \cref{eq:rescaling_d} \textit{et seq.}].

\subsection{Two-dimensional honeycomb networks\label{sec:honeycomb}}
We generated periodic unit cells of honeycomb networks having $N_{\mathrm{row}}$ rows of hexagons as illustrated in \cref{fig:honeycomb_illustration} for the case of $N_{\mathrm{row}}=5$. We used uniform vessel lengths ($l=1$ in nondimensional units) and a diameter distribution as described above [see \cref{eq:log_normal}].

We ran $10^6$ realizations, in each of which we evaluated the macroscopic \linebreak
permeability [using \cref{eq:K_final}] of a honeycomb network ($N_{\mathrm{row}}=9$) with vessel diameters chosen at random using the log-normal distribution defined by \cref{eq:log_normal}; each realization had a fixed
haematocrit in all vessels in the range $h \in [0,0.7]$. The results presented in \cref{fig:K_scatter} show how the diagonal and off-diagonal entries of the permeability tensor are
distributed. Due to the intrinsic symmetries of the honeycomb network, the
distribution of the diagonal components is identical. Therefore, \cref{fig:K_scatter}(a) presents the distribution of $\mathcal{K}^{11}$ and $\mathcal{K}^{22}$, which we denote as $\mathcal{K}^{nn}$. Also, due to the symmetries of the honeycomb-network, the mean value of the off-diagonal components (denoted $\mathcal{K}^{nt}$, where $\mathcal{K}^{nt}=\mathcal{K}^{tn}$ due to the symmetry of the permeability tensor) is zero. This is readily apparent in \cref{fig:K_scatter}(b).

Curves presenting the permeability tensor components as a function of $h$ for the case of uniform vessel diameters are presented in \cref{fig:K_scatter}. Due to the symmetry of the honeycomb network, the uniform vessel diameter results in $\langle\mathcal{K}^{nt}\rangle=0$. For the diagonal terms, \cref{fig:K_scatter}(a) shows that a uniform diameter which equals the mean of the distribution ($\widetilde{d}=20$) produced a larger permeability than the mean of the distribution. This is due to the nonlinear dependence of the vessel conductance on the vessel diameter, which leads to vessels smaller than the mean value producing a greater reduction in the permeability than increases caused by vessels that are larger than the mean value. Even so, it is possible to obtain a good approximation for the mean value by using a uniform value for the vessel diameters which is smaller than the mean of the distribution. 

\begin{figure}[H]
\includegraphics[scale=0.45]{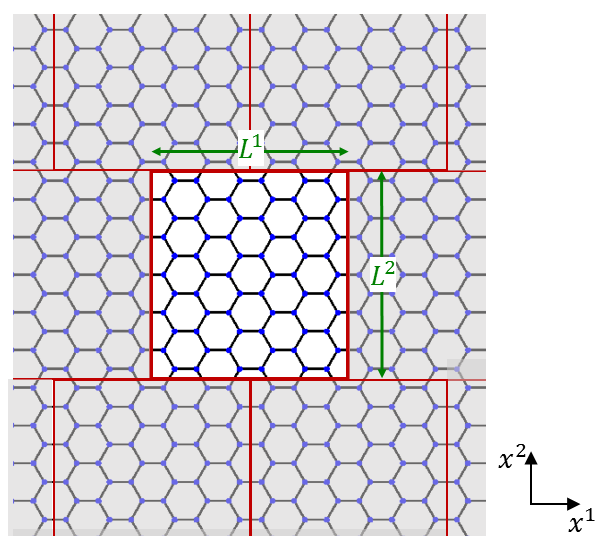}
\centering{}\caption{Schematic of a 2D honeycomb network in the $\left(x^1,x^2\right)$ plane composed of unit cells having $N_{\mathrm{row}}=5$ rows of hexagons. Each unit cell has six neighbouring cells corresponding to (clockwise from the top-left cell) $\mathbf{r}=\left(-\frac{1}{2},1\right),\left(\frac{1}{2},1\right),\left(1,0\right),\left(\frac{1}{2},-1\right),\left(-\frac{1}{2},-1\right),\left(-1,0\right)$.
\label{fig:honeycomb_illustration}}
\end{figure}

\begin{figure}[H]
\includegraphics[scale=0.26]{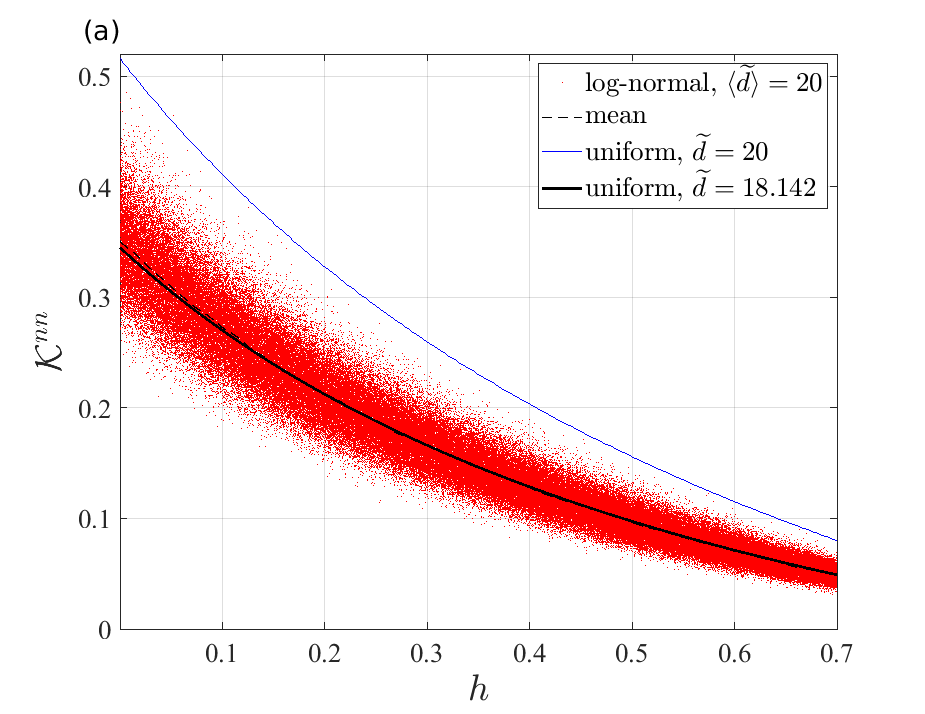}\hfill{}\includegraphics[scale=0.26]{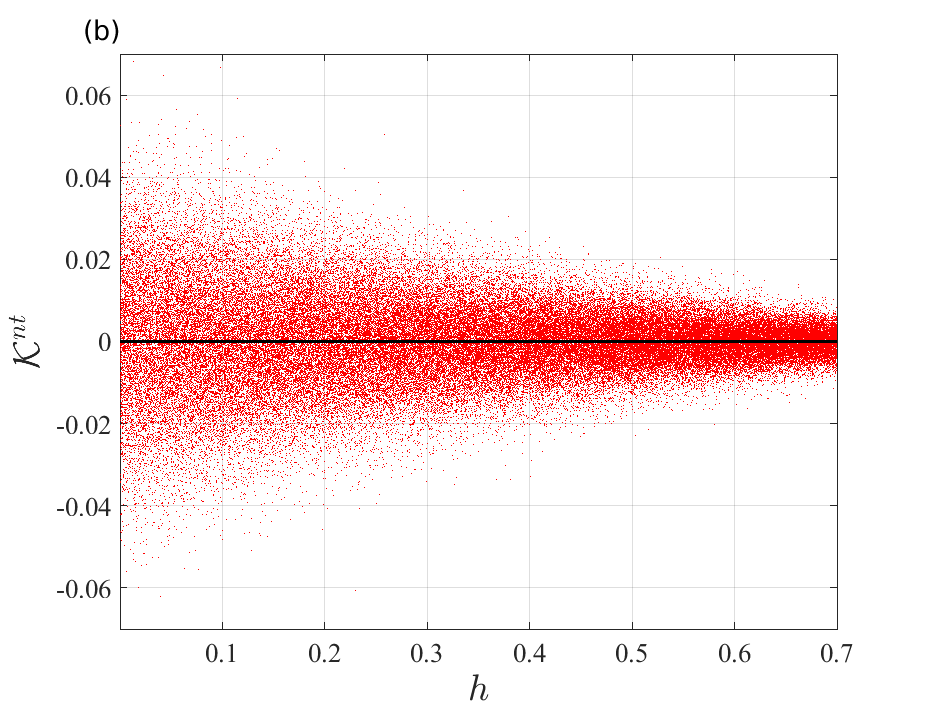}
\centering{}\caption{Scatter plots showing how the diagonal and off-diagonal component elements of the permeability tensor, (a)~$\mathcal{K}^{nn}$ and (b)~$\mathcal{K}^{nt}$, respectively, depend on the macroscopic haematocrit, $h$. The results were generated from $10^6$ realizations of a honeycomb network in which the unit cell had $N_{\mathrm{row}}=9$ rows (190 nodes, 304 vessels) and where the vessel diameters were distributed according to a log-normal distribution with mean diameter, $\langle\widetilde{d}\rangle=20$, and $\sigma_{LN}=0.3$. The solid lines show the permeability for cases in which uniform vessel diameters were used; the mean of the distribution of the permeability is given by the dashed-black line [coincides with the uniform-vessel-diameter line in \cref{fig:K_scatter}(b)].
\label{fig:K_scatter}}
\end{figure}

When evaluating the macroscopic permeability (per unit length) of a honeycomb unit-cell having uniform vessel diameters, we can take the conductance out of the summation in \cref{eq:K_final} to obtain
\begin{equation}
\mathcal{K}^{kl}=\frac{L^{k}L^{l}}{2L^{1}L^{2}}\frac{(\widetilde{d}/\widetilde{d}_0)^{4}}{\mu(h,\widetilde{d})}\sum_{\mathbf{r} \in \mathcal{N}}\sum_{i=1}^{N}\sum_{j=1}^{N} \left(r^{k}r^{l}+2r^{k}v_{i}^{l}\right).
\label{eq:K_uniform_d}
\end{equation}
Here we choose the reference diameter to be the mean diameter so that the \linebreak
nondimensional parameter, $\widetilde{d}_0=20$, represents the reference network diameter scaled by the reference viscosity-function diameter. The vessel length is absent from \cref{eq:K_uniform_d} because it is uniform in honeycomb networks.  For an equal-vessel-diameter \linebreak
honeycomb network, the summation in \cref{eq:K_uniform_d} yields
\begin{equation}
\mathcal{K}^{nn} \approx 0.577\frac{(\widetilde{d}/\widetilde{d}_0)^{4}}{\mu(h,\widetilde{d})} \ \ \text{and} \ \ \mathcal{K}^{nt}=0.
\label{K_uniform_honeycomb}
\end{equation}
On the other hand, if we wish to approximate the mean of the distribution of $\mathcal{K}^{nn}$ using a representative uniform diameter, we expect that
\begin{equation}
\langle \mathcal{K}^{nn} \rangle \propto \left[\mu(h,\widetilde{d}=\mathrm{const})\right]^{-1}.
\end{equation}
Using the explicit expression for $\mu$ (see \cref{app:blood_viscosity}), we fitted the following functional form for $\langle \mathcal{K}^{nn}(h) \rangle$ to the random-diameter data in \cref{fig:K_scatter},
\begin{equation}
\langle \mathcal{K}^{nn}(h) \rangle = \frac{c_1}{1+c_2 \left[(1-h)^{c_3}-1 \right]},
\label{K_fit}
\end{equation}
and obtained $c_1\approx0.345$, $c_2\approx2.359$, and $c_3\approx-1.059$. Then, the representative uniform diameter of the random distribution is obtained by substituting the explicit form for $\mu(h,\widetilde{d})$ from \cref{eq:blood_viscosity} into \cref{K_uniform_honeycomb} and equating to \cref{K_fit} for the case of $h=0$. The resulting equation is
\begin{equation}
\widetilde{d}^2(\widetilde{d}-1.1)^2=\frac{0.345\widetilde{d}^4_0}{0.577}.
\label{eq:representative_diameter}
\end{equation}
Solving \cref{eq:representative_diameter} numerically we obtained $\widetilde{d} \approx 18.142$.
The solid black line in \cref{fig:K_scatter}(a) shows $\mathcal{K}^{nn}$ for a uniform diameter, $\widetilde{d}=18.142$, while the dashed-black line presents the mean, $\langle\mathcal{K}^{nn}\rangle$, of the random vessel diameter distribution. The two curves are almost indistinguishable, showing that this uniform diameter provides an excellent approximation to the mean of the distribution.

It is noticeable in \cref{fig:K_scatter} that the standard deviation of the permeability decreases as the haematocrit value increases. To better illustrate this trend, \linebreak \cref{fig:probability}(a,b) show the probability density distribution of the permeability tensor components for two values of haematocrit, $h=0.1,0.6$, while \cref{fig:probability}(c) presents the standard deviation and the coefficient of variation of the permeability tensor components as a function of the haematocrit value. We notice in \cref{fig:probability}(a,b) the wider distribution of the permeability components for the smaller value of \linebreak haematocrit. This is also manifested in \cref{fig:probability}(c) by the decrease in the standard deviation of the distributions as the haematocrit increases. However, the coefficient of variation of $\mathcal{K}^{nn}$ increases as the haematocrit increases, meaning that we expect to observe larger relative deviations from the mean value for larger haematocrit values. This is because the gradients of the blood viscosity  with respect to the vessel diameter increase when the haematocrit increases (see, for example, Figure.~1.12 in \cite{pries2008blood}). This, in turn, leads to greater relative variability in vessel conductance for larger values of haematocrit.

\begin{figure}[H]
\includegraphics[scale=0.25]{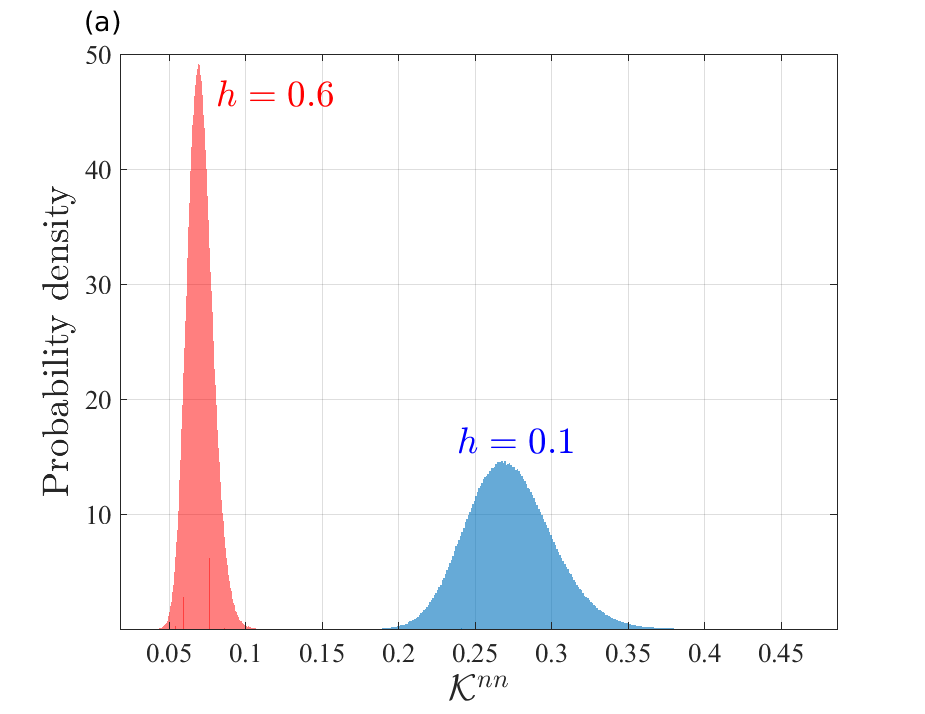}\hspace*{0.5cm}\includegraphics[scale=0.25]{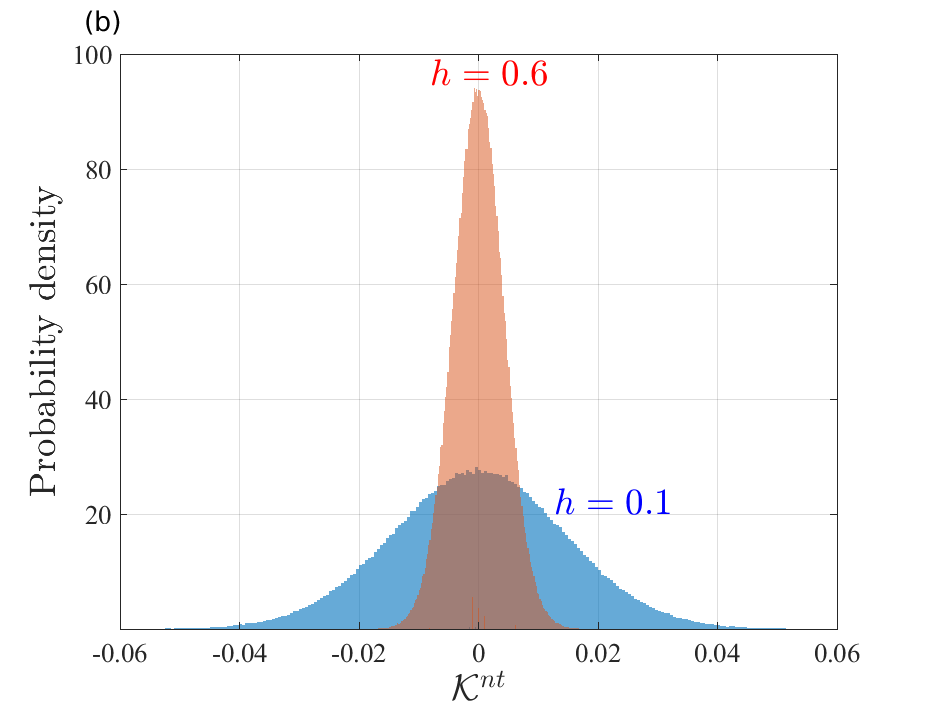}
\centering{}
\includegraphics[scale=0.25]{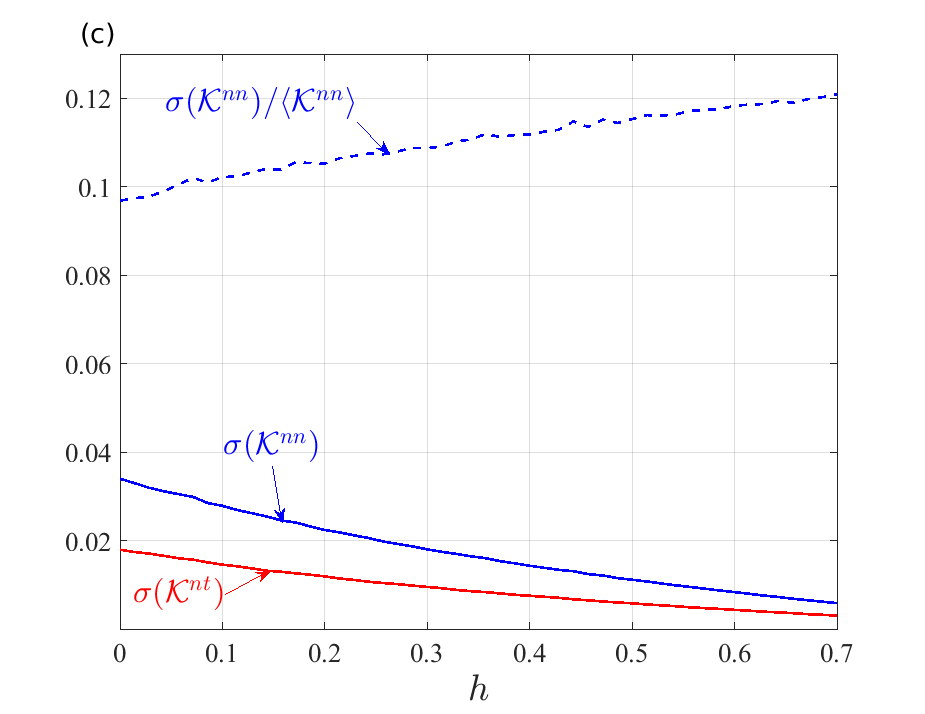}
\caption{The probability density distribution of (a) $\mathcal{K}^{nn}$ and (b) $\mathcal{K}^{nt}$ for haematocrit values of $h=0.1$ (blue bars) and $h=0.6$ (red bars). The results were generated from $10^6$ realizations of a honeycomb network in which a periodic unit cell had $N_{\mathrm{row}}=9$ rows (190 nodes, 304 vessels) and where the vessel diameters are distributed according to a log-normal distribution with  mean diameter, $\langle \widetilde{d}\rangle=20$, and $\sigma_{LN}=0.3$. In \cref{fig:probability}(c) the standard deviation of the distribution, $\sigma \left(\mathcal{K}^{nn}\right)$ (solid-blue line) and $\sigma\left(\mathcal{K}^{nt}\right)$ (solid-red line), are plotted against the haematocrit value; the coefficient of variation of $\mathcal{K}^{nn}$ is shown by the dashed-blue line. \label{fig:probability}}
\end{figure}

\subsection{Three-dimensional networks with random vessel orientations\label{sec:random_3D}}

We used the following algorithm to generate three-dimensional unit cells containing \linebreak
periodic networks with random vessel orientations:

\begin{enumerate}[label=(\Roman*)]
\item We defined a cubical cell, with an edge length $\widetilde{L}_0$, that contains the unit-cell nodes (as usual, tildes denote scaling with $1\, \unit{\mu m}$).

\item We generated $N_{\mathrm{tree}}$ realizations of tree-type networks where each tree had $N_{\mathrm{gen}}$ generations. For each tree we positioned the first node at a random location within the unit cell and then randomly chose the positions of the two daughter nodes such that the distance of each node from the parent node (i.e., the vessel length) was sampled according to a given vessel length distribution [see item (III) below]. The position of each node was sampled from a uniform distribution over the surface of a sphere where its origin was the parent node and its radius was the sampled vessel length. We repeated the process for the daughter nodes (now viewing them as the parent nodes of the next generation) until the required number of generations, $N_{\mathrm{gen}}$, was reached.

\item The distribution of the vessel length-to-diameter ratio, $R$, was chosen from a log-normal distribution with a mean $\langle R \rangle$. Assigning the random log-normally distributed vessel diameter, $\widetilde{d}$, as in \cref{eq:log_normal}, the vessel length, $\widetilde{l}$, was given by $\widetilde{l}=R\widetilde{d}$.

\item In order to impose periodicity, if the position $\mathbf{\widetilde{x}_0}$ of a node fell outside the cell box, it was shifted according to the transformation
\[
\widetilde{x}_0^{k} \rightarrow \begin{cases}
\widetilde{x}_0^{k} + \widetilde{L}_0, & \widetilde{x}_0^{k} < -\widetilde{L}_0/2 \\
\widetilde{x}_0^{k} - \widetilde{L}_0, & \widetilde{x}_0^{k} > \widetilde{L}_0/2 \\
\end{cases}
\]

\item We randomly connected the open ends in the first and last generations of each tree between the different trees or within the same tree such that each node in the network had a coordination number of 3.

\item The nondimensional components of the permeability tensor were calculated via
\begin{equation}
\mathcal{K}^{kl}=\frac{\langle\widetilde{l}\rangle^2}{2\langle\widetilde{d}\rangle^4\widetilde{L}_0}\sum_{\mathbf{r} \in \mathcal{N}}\sum_{i=1}^{N}\sum_{j=1}^{N}\frac{ \widetilde{d}_{ij}^{4}\left(r^{k}r^{l}+2r^{k}v_{i}^{l}\right)}{\widetilde{l}_{ij}\mu(h,\widetilde{d}_{ij})}.
\label{eq:K_3D}
\end{equation}

\end{enumerate}

An example of a periodic unit cell generated using the aforementioned algorithm is presented in \cref{fig:3d_random_illustration}.
\begin{figure}[H]
\includegraphics[scale=0.4]{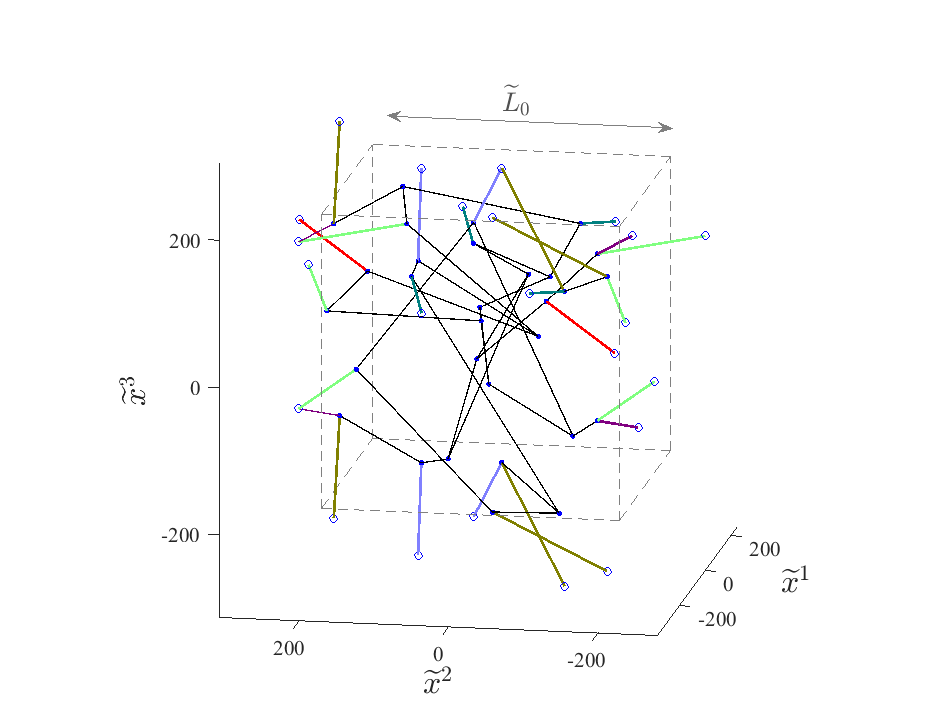}
\centering{}\caption{Illustration of a periodic unit cell of a three-dimensional random network composed of $N_{\mathrm{tree}}=2$ trees, each having $N_{\mathrm{gen}}=3$ generations. Vessels connected to neighbouring unit cells are marked in colour, where vessels connected to the same or opposite neighbour  (having identical $|\bf r|$) share the same colour. The cell box is illustrated in dashed-grey lines; nodes outside of the cell box belonging to neighbouring unit cells are indicated by open circles.
\label{fig:3d_random_illustration}}
\end{figure}

We ran $10^6$ realizations of 3D periodic unit cells (as described above) each with uniform haematocrit values in the range $h \in [0,0.7]$. \Cref{fig:K_scatter_3d_random} shows the scatter in the components of the resulting permeability tensor. Due to the average anisotropy induced by the randomly oriented vessels in the network, the distributions of all diagonal components are similar, as are the distributions of all off-diagonal \linebreak
components. Therefore, we present the distribution of all diagonal terms together and denote them by $\mathcal{K}^{nn}$ [\cref{fig:K_scatter_3d_random}(a)]; the off-diagonal terms are denoted by \linebreak
$\mathcal{K}^{nt}$ [\cref{fig:K_scatter_3d_random}(b)]. Here, because the network is anisotropic on average, the mean of the off-diagonal terms, $\langle\mathcal{K}^{nt}\rangle=0$, which is readily apparent in \cref{fig:K_scatter_3d_random}(b). 
Comparing the order-of-magnitude of the permeabilities of the 3D random network and the 2D honeycomb network (\cref{sec:honeycomb}), we notice that the former is approximately ten times larger. This means that a three-dimensional network consisting of 2D honeycomb unit cells lying in the $(x^1,x^2)$ plane stacked in the $x^3$ direction, with a density $\sim 10/l^*_0$ (where $l^*_0$ is the dimensional vessel length in the honeycomb unit cell) will have a permeability similar in magnitude to that of the 3D network presented in \cref{fig:K_scatter_3d_random}.

\begin{figure}[H]
\includegraphics[scale=0.26]{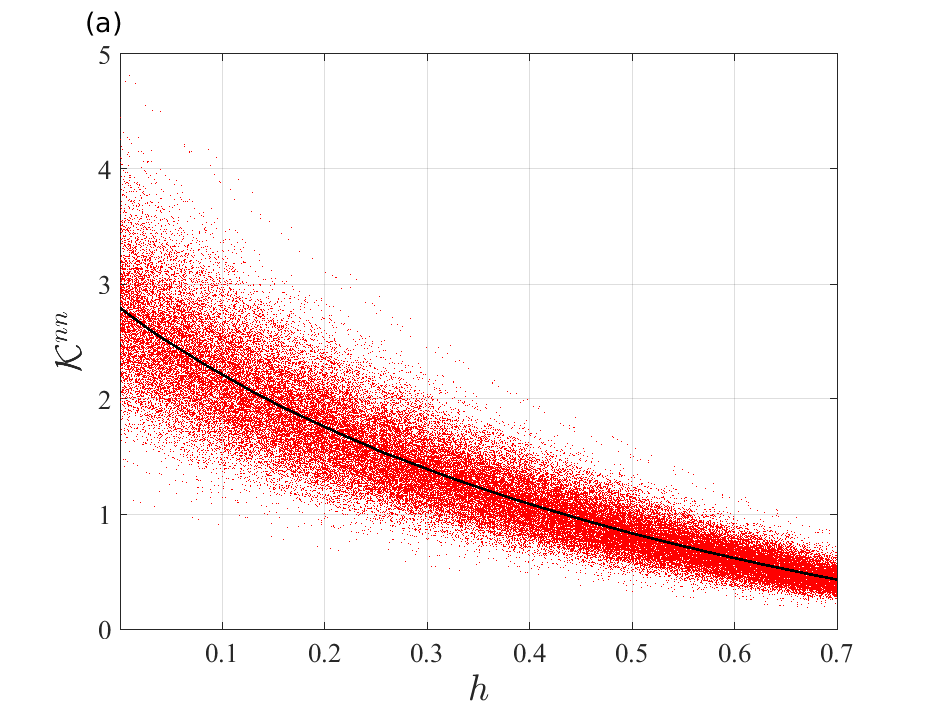}\hfill{}\includegraphics[scale=0.26]{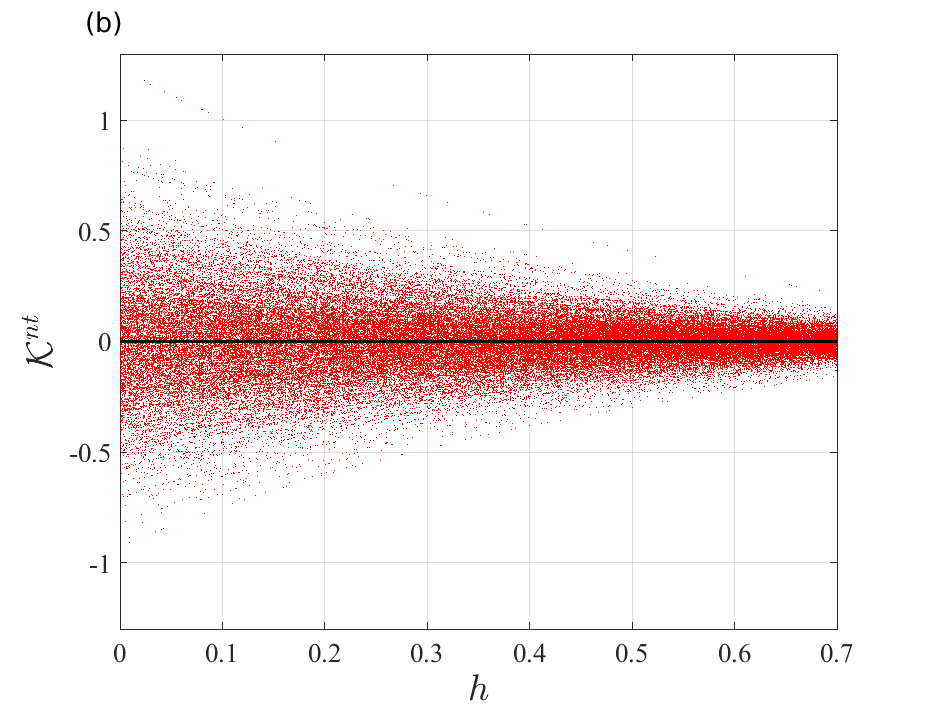}
\centering{}\caption{Scatter plots showing how the diagonal and off-diagonal component elements of the permeability tensor, (a)~$\mathcal{K}^{nn}$ and (b)~$\mathcal{K}^{nt}$, respectively, depend on the haematocrit, $h$. The results were generated from $10^6$ realizations of a 3D random network in which a periodic unit cell had $N_{\mathrm{tree}}=6$ trees, each composed of $N_{\mathrm{gen}}=5$ generations (378 nodes, 567 vessels), in a cubical box with edge length $\widetilde{L}_0=400$. The vessel diameters are distributed according to a log-normal distribution with mean diameter, $\langle\widetilde{d}\rangle=20$, and $\sigma_{LN}=0.3$; the length-to-diameter ratio of the vessels is log-normally distributed with mean, $\langle R \rangle=6$, and $\sigma_{LN}=0.3$. The means of the distributions are given by the black lines.
\label{fig:K_scatter_3d_random}}
\end{figure}

To illustrate how the statistical features of the 3D random networks change as the haematocrit value varies, in \cref{fig:probability_3d_random}(a,b) we plot the probability density distribution of the permeability tensor components for two values of haematocrit, $h=0.1,0.6$; \cref{fig:probability_3d_random}(c) presents the results for the standard deviation and coefficient of variation of the permeability tensor components as a function of the haematocrit value. As in \cref{fig:probability}, the standard deviation of the distributions decreases as the haematocrit increases, while the coefficient of variation of $\mathcal{K}^{nn}$ increases as the haematocrit increases. However, the relative increase in the coefficient of variation is smaller in the 3D random network than in the honeycomb network ($\approx 10\%$ increase in the former compared to $\approx 20\%$ increase in the latter between $h=0$ to $h=0.7$). This is because the 3D random network allows for a change in the vessel length, such that vessels with larger diameter are more likely to be longer. Due to the opposing effects of vessel diameter and length on the conductance [see \cref{{eq:K_final}}], this leads to smaller relative variability in the vessel conductances as the haematocrit increases compared to the case in which the vessel lengths are fixed.

\begin{figure}[H]
\centering{}
\includegraphics[scale=0.25]{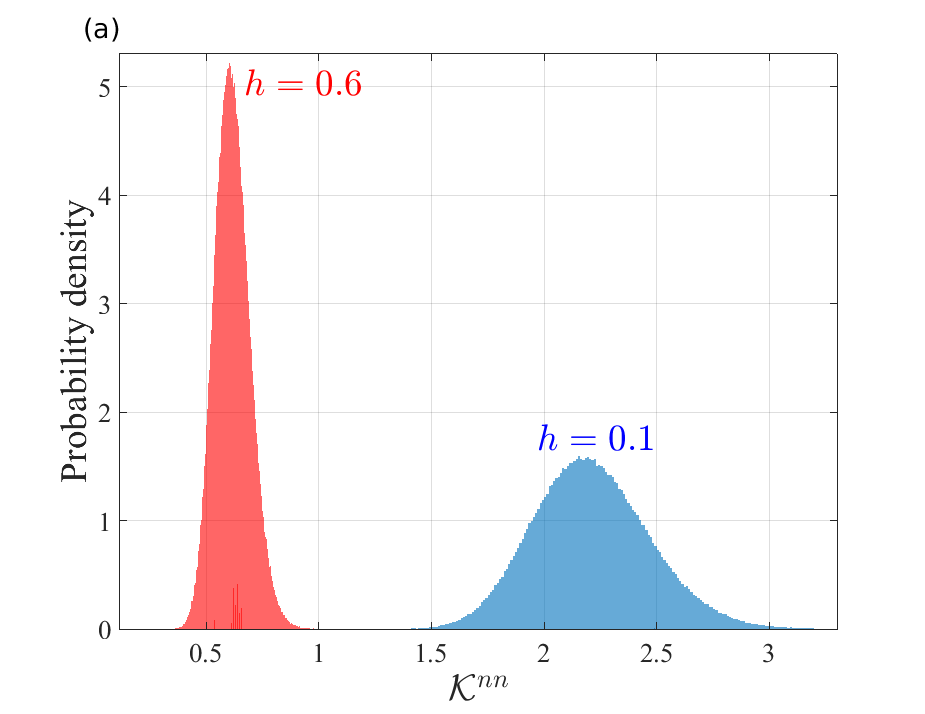}\hspace*{0.5cm}\includegraphics[scale=0.25]{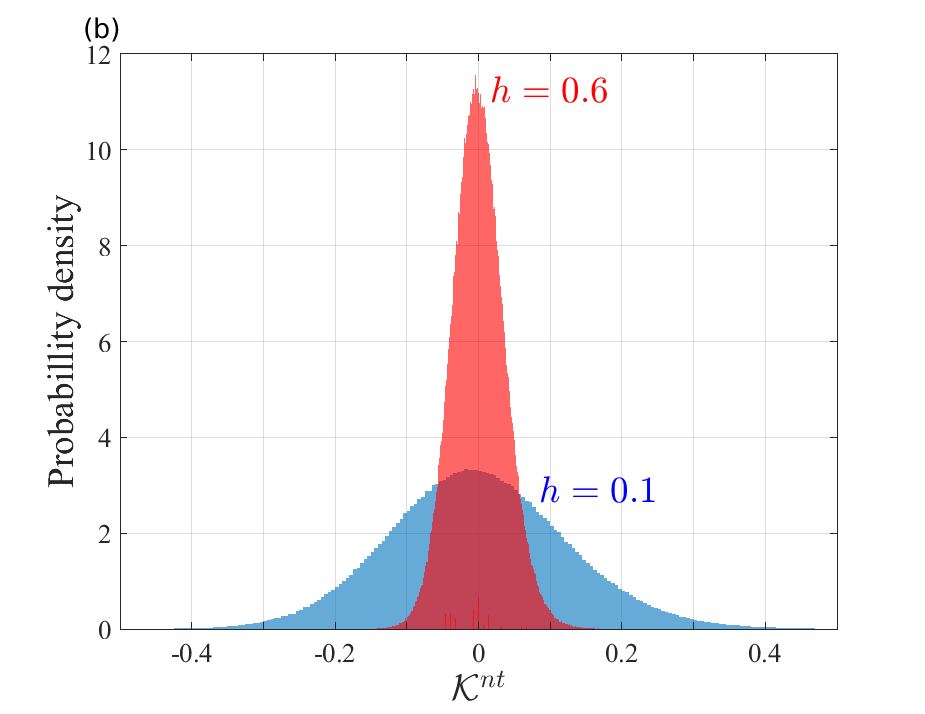}
\centering{}
\includegraphics[scale=0.25]{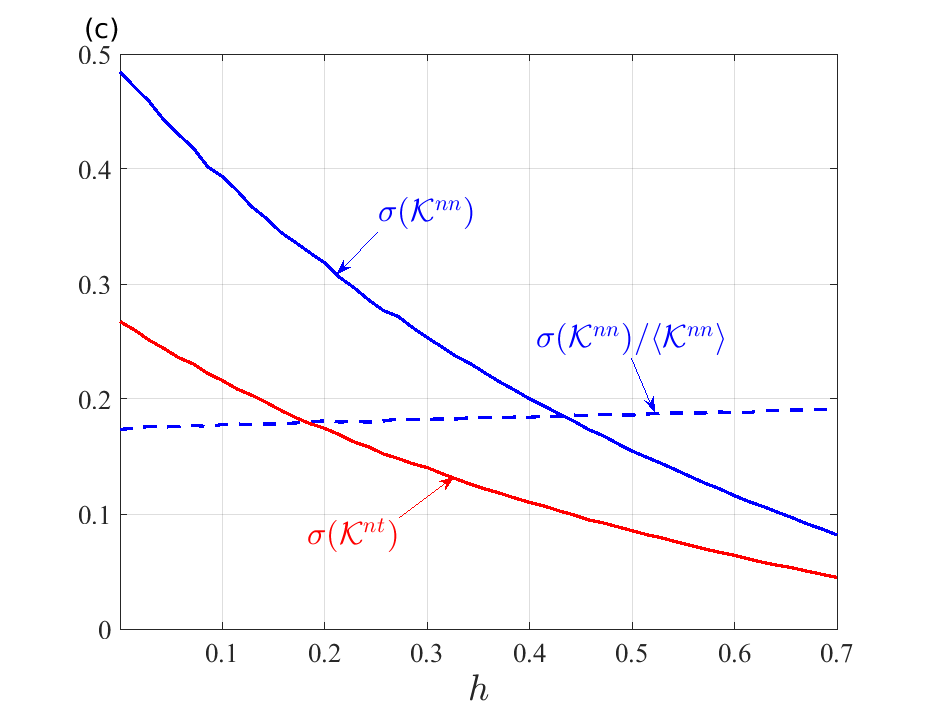}
\caption{The probability density distribution of (a) $\mathcal{K}^{nn}$ and (b) $\mathcal{K}^{nt}$ for haematocrit values of $h=0.1$ (blue bars) and $h=0.6$ (red bars) generated from $10^6$ realizations of a 3D random network in which a periodic unit cell had $N_{\mathrm{tree}}=6$ trees, each composed of $N_{\mathrm{gen}}=5$ generations (378 nodes, 567 vessels), in a cubical box with edge length $\widetilde{L}_0=400$. The vessel diameters are distributed according to a log-normal distribution with mean diameter, $\langle\widetilde{d}\rangle=20$, and $\sigma_{LN}=0.3$; the length-to-diameter ratio of the vessels is log-normally distributed with mean, $\langle R \rangle=6$, and $\sigma_{LN}=0.3$. In \cref{fig:probability_3d_random}(c) the standard deviation of the distribution, $\sigma \left(\mathcal{K}^{nn}\right)$ (solid-blue line) and $\sigma\left(\mathcal{K}^{nt}\right)$ (solid-red line), are plotted against the haematocrit value; the coefficient of variation of $\mathcal{K}^{nn}$ is shown by the dashed-blue line. \label{fig:probability_3d_random}}
\end{figure}

\section{Conclusion\label{sec:conclusion}}
The goal of the present work was to formulate a macroscopic, coarse-grained description for blood flow in the microcirculation. We adopted a standard description of blood as a continuous mixture of plasma and haematocrit, which is formulated as a nonlinear system of algebraic equations for the nodal \linebreak
pressures, flow and haematocrit in the network edges (vessels). Here, two factors drive the nonlinearities: the dependence of blood viscosity on the haematocrit \cite{pries1994resistance}, and the coupling of haematocrit and flow at network bifurcations due to the nonlinear splitting phenomenon \cite{pries1989red}. In contrast to previous studies of flow in networks \cite{chapman2021} and blood flow homogenisation (e.g., \cite{chapman2008multiscale,el2015multi,shipley2010multiscale}), which assumed a single-phase, linear Stokes flow, in this work we accounted for the two-phase nature of blood and the associated nonlinearities. Consequently, we have developed a coupled set of equations describing the flow and haematocrit propagation at the macroscale.

In order to account for the nonlinearity of blood flow and to facilitate \linebreak
the homogenisation procedure, we focused on the limit of small haematocrit \linebreak
heterogeneity. This scenario prevails when the diameter of the vessels in the network are sufficiently large (typically much larger than the diameter of a single RBC), and when the diameters of daughter vessels at bifurcations are not markedly different. Under these conditions we were able to study the asymptotic limit at which the vessel haematocrits deviate slightly from the average haematocrit in the periodic unit cell and, thus, to formulate a coarse-grained equation for the average-haematocrit propagation. Interestingly, but perhaps not surprisingly, this resulted in a purely advective equation for the propagation of haematocrit at the macroscale, where the effect of haematocrit splitting has been averaged out. Considering the homogenisation of the flow, in this limit, the slow change of the leading order haematocrit enabled us to follow the homogenisation methodology developed by \chap and to derive a Darcy model with a haematocrit-dependent permeability tensor.

We calculated the components of the permeability tensor for two different network topologies: (i) 2D honeycomb networks and (ii) 3D networks with randomly oriented vessels. For both types of networks a distribution of vessel diameters was considered. As expected, the mean permeability decreases as the haematocrit value increases due to the increase in flow resistance. However, it was observed that the coefficient of variation of the permeability increases with the haematocrit. This was rationalized in terms of the increase in the blood viscosity gradient with respect to the vessel diameter as the haematocrit level increases.

The homogenisation methodology used in the paper is based on the assumption that the network is homogeneous at the microscale such that it can be approximated using a representative unit cell (or at least a statistical distribution of this unit cell). This assumption is reasonable for the more distal parts of the microcirculation, while the effect of larger vessels cannot be accounted for within the homogenisation methodology. Peyrounette et al. \cite{peyrounette_lorthois_2018}, and later on Shipley et al. \cite{shipley2020hybrid}, derived hybrid discrete-continuum models for blood flow where a homogenised model is used in the capillary region and the larger vessels, feeding the capillaries, are solved discretely. While in the models by \cite{peyrounette_lorthois_2018,shipley2020hybrid} the haematocrit value was fixed, they can now be extended to account for spatially varying haematocrit using the homogenised model developed in this paper.

The limit of small haematocrit heterogeneity studied in this work is likely to prevail in regions of the microcirculation where vessel diameters are sufficiently large and when vessel bifurcations are mostly regular (bifurcations of vessels with similar diameters). However, tumour vascular networks are known to be highly abnormal and, thus, to promote irregular blood flow and haematocrit distributions \cite{Jain05,nagy2009tumour}. Understanding this irregular flow is of great biological significance as it is hypothesised to give rise to the phenomenon of cycling hypoxia in tumours \cite{Gil18,CTF16} where some regions of the tumour experience cycles of hypoxia and re-oxygenation. In turn, cycling hypoxia is associated with poor prognosis, such as selective advantage for malignant growth \cite{bottaro2003out,Hoc96}, and chemo- and radio-therapy resistance \cite{Gra53,HB04,Hor12}.
In order to homogenise the blood flow and haematocrit transport in tumours, the small haematocrit heterogeneity assumption must be relaxed and the nonlinearity of blood-flow at the microscale taken into account. Additionally, to account for the \linebreak
phenomenon of cycling hypoxia, a methodology to homogenise unsteady blood-flow must be developed.
While this may not be analytically tractable, a data-driven approach where the unit-cell problem is solved numerically to generate data for the coarse-graining process \cite{taghizadeh2022} may prove useful. This suggests an avenue for future research.

Since RBCs are the blood component that is responsible for transporting and delivering oxygen to biological tissues, it is important to better understand their distribution in the microcirculation. Having formulated equations that describe the transport of haematocrit at the macroscale, it is now possible to derive a coarse-grained description for the oxygen distributions at the tissue-level. This is especially important in order to understand how large-scale regions of (sustained) hypoxia form in tumours and, thus, constitutes a natural topic for future work.

\vspace{5mm}

\appendix
\section{The blood viscosity model\label{app:blood_viscosity}}
Here we detail the functional form of the empirical model for the \emph{apparent} blood viscosity due to Pries et al. \cite{pries1994resistance}.
For a vessel having a diameter, $d$ (in micron units), and haematocrit fraction, $h$, the relative viscosity (blood apparent viscosity divided by plasma viscosity) is given by
\begin{align}
    \mu(h,d) &= \left[1+(\mu^\circ-1)\frac{(1-h)^C-1}{(1-0.45)^C-1} \left(\frac{d}{d-1.1}\right)^2 \right]\nonumber \\
    &\times \left(\frac{d}{d-1.1}\right)^2.
    \label{eq:blood_viscosity}
\end{align}
The following expression for the blood viscosity is fit to data collected \emph{in vivo} by \cite{pries1994resistance}. Here, the reference viscosity, $\mu^\circ$, and the exponent, $C$, are given by the following empirical expressions 
\begin{equation}
    \mu^\circ(d) = 6 \exp\left(-0.085 d\right)-2.44 \exp\left(-0.06 d^{0.645} \right) +3.2
\end{equation}
and
\begin{align}
C(d)=& \left[0.8+\exp(-0.075 d) \right] 
\left\{ \left[1+10\left(\frac{d}{10}\right)^{12} \right]^{-1}-1\right\}\nonumber \\
&+\left[1+10\left(\frac{d}{10}\right)^{12} \right]^{-1}.
\end{align}

\section{Proofs for the identities in \cref{eq:identities}\label{app:identities_1}}

\begin{equation} 
\begin{aligned}
A_{0}^{T}= & \sum_{\mathbf{r} \in \mathcal{N}}\left[g_{ji}^{\mathbf{r}}-\delta_{ji}\sum_{m=1}^{N}g_{jm}^{\mathbf{r}}\right]=\sum_{\mathbf{r} \in \mathcal{N}}\left[g_{ji}^{-\mathbf{r}}-\delta_{ji}\sum_{m=1}^{N}g_{jm}^{\mathbf{r}}\right]\\
= & \sum_{\mathbf{r} \in \mathcal{N}}\left[g_{ij}^{\mathbf{r}}-\delta_{ji}\sum_{m=1}^{N}g_{jm}^{\mathbf{r}}\right] = A_{0},
\end{aligned}
\end{equation}
where we have relabeled $\mathbf{r}\leftrightarrow-\mathbf{r}$ in
the first sum and then used relation (\ref{eq:g_period}).

\begin{equation} 
\left(A_{0}\mathbf{u}\right)_{i}=\sum_{\mathbf{r} \in \mathcal{N}}\sum_{j=1}^{N}\left[g_{ij}^{\mathbf{r}}-\delta_{ij}\sum_{m=1}^{N}g_{im}^{\mathbf{r}}\right]=\sum_{\mathbf{r} \in \mathcal{N}}\left[\sum_{j=1}^{N}g_{ij}^{\mathbf{r}}-\sum_{m=1}^{N}g_{im}^{\mathbf{r}}\right]=\mathbf{0},\;\text{for}\;i=1,..,N.
\end{equation}
In a similar manner we can show that $A_{1}\mathbf{u}=\mathbf{0}$.

\begin{equation}
\begin{aligned}\left(\mathbf{u}^{T}A_{0}\right)_{j}= & \sum_{\mathbf{r} \in \mathcal{N}}\sum_{i=1}^{N}\left[g_{ij}^{\mathbf{r}}-\delta_{ij}\sum_{m=1}^{N}g_{im}^{\mathbf{r}}\right]=\sum_{\mathbf{r} \in \mathcal{N}}\left[\sum_{i=1}^{N}g_{ij}^{\mathbf{r}}-\sum_{m=1}^{N}g_{jm}^{\mathbf{r}}\right]\\
= & \sum_{\mathbf{r} \in \mathcal{N}}\left[\sum_{i=1}^{N}g_{ij}^{\mathbf{r}}-\sum_{m=1}^{N}g_{jm}^{\mathbf{-r}}\right]\\
= & \sum_{\mathbf{r} \in \mathcal{N}}\left[\sum_{i=1}^{N}g_{ij}^{\mathbf{r}}-\sum_{m=1}^{N}g_{mj}^{\mathbf{r}}\right]=\mathbf{0},\;\;\;\text{for}\;j=1,..,N.
\end{aligned}
\end{equation}
where we have relabeled $\mathbf{r}\leftrightarrow-\mathbf{r}$ in
the second sum and then used relation (\ref{eq:g_period}).

\begin{equation}
\begin{aligned}
\mathbf{u}^{T}B_{0}^{l}\mathbf{u}= & \sum_{\mathbf{r} \in \mathcal{N}}\sum_{i=1}^{N}\sum_{j=1}^{N}\left[(x_{j}^{l}+r^{l}L^{l})g_{ij}^{\mathbf{r}}-\delta_{ij}x_{i}^{l}\sum_{m=1}^{N}g_{im}^{\mathbf{r}}\right]\\
= & \sum_{\mathbf{r} \in \mathcal{N}}\sum_{i=1}^{N}\sum_{j=1}^{N}\left[(x_{i}^{l}+r^{l}L^{l})g_{ji}^{\mathbf{r}}-\delta_{ij}x_{i}^{l}\sum_{m=1}^{N}g_{im}^{\mathbf{r}}\right]\\
= & \sum_{\mathbf{r} \in \mathcal{N}}\sum_{i=1}^{N}\sum_{j=1}^{N}\left[(x_{i}^{l}-r^{l}L^{l})g_{ji}^{-\mathbf{r}}-\delta_{ij}x_{i}^{l}\sum_{m=1}^{N}g_{im}^{\mathbf{r}}\right]\\
= & \sum_{\mathbf{r} \in \mathcal{N}}\left[\sum_{i=1}^{N}(x_{i}^{l}-r^{l}L^{l})\sum_{j=1}^{N}g_{ij}^{\mathbf{r}}-\sum_{i=1}^{N}x_{i}^{l}\sum_{m=1}^{N}g_{im}^{\mathbf{r}}\right]\\
= & -L^{l}\sum_{\mathbf{r} \in \mathcal{N}}\sum_{i=1}^{N}\sum_{j=1}^{N}r^{l}g_{ij}^{\mathbf{r}}=L^{l}\sum_{\mathbf{r} \in \mathcal{N}}\sum_{i=1}^{N}\sum_{j=1}^{N}r^{l}g_{ji}^{\mathbf{-r}}=L^{l}\sum_{\mathbf{r} \in \mathcal{N}}\sum_{i=1}^{N}\sum_{j=1}^{N}r^{l}g_{ij}^{\mathbf{r}}=0,
\end{aligned}
\end{equation}
where we have relabeled $i\leftrightarrow j$ and $\mathbf{r}\leftrightarrow-\mathbf{r}$
in the first sum, and then used (\ref{eq:g_period}). We followed
the same procedure in the fifth line.

\vspace{5mm}

\section{Proofs for the identities in \cref{eq:identities_2}\label{app:identities_2}}

\begin{equation}
\begin{aligned}
& \mathbf{u}^{T}A_{1}(g_{ij}^{\mathbf{r}(0)})\mathbf{w}^{l} =\\
&  \frac{1}{2}\sum_{\mathbf{r} \in \mathcal{N}}\sum_{j=1}^{N}w_{j}^{l}\left[\sum_{i=1}^{N}\sum_{s=1}^{3}(x_{i}^{s}+x_{j}^{s}+r^{s}L^{s})\frac{\partial g_{ij}^{\mathbf{r}(0)}}{\partial x^{s}}-\sum_{m=1}^{N}\sum_{s=1}^{3}(x_{j}^{s}+x_{m}^{s}+r^{s}L^{s})\frac{\partial g_{jm}^{\mathbf{r}(0)}}{\partial x^{s}}\right]\\
& = \frac{1}{2}\sum_{\mathbf{r} \in \mathcal{N}}\sum_{j=1}^{N}w_{j}^{l}\left[\sum_{i=1}^{N}\sum_{s=1}^{3}(x_{i}^{s}+x_{j}^{s}+r^{s}L^{s})\frac{\partial g_{ij}^{\mathbf{r}(0)}}{\partial x^{s}}-\sum_{m=1}^{N}\sum_{s=1}^{3}(x_{j}^{s}+x_{m}^{s}-r^{s}L^{s})\frac{\partial g_{jm}^{\mathbf{-r}(0)}}{\partial x^{s}}\right]\\
& = \frac{1}{2}\sum_{\mathbf{r} \in \mathcal{N}}\sum_{j=1}^{N}w_{j}^{l}\left[\sum_{i=1}^{N}\sum_{s=1}^{3}(x_{i}^{s}+x_{j}^{s}+r^{s}L^{s})\frac{\partial g_{ij}^{\mathbf{r}(0)}}{\partial x^{s}}-\sum_{m=1}^{N}\sum_{s=1}^{3}(x_{j}^{s}+x_{m}^{s}-r^{s}L^{s})\frac{\partial g_{mj}^{\mathbf{r}(0)}}{\partial x^{s}}\right]\\
& =  \sum_{\mathbf{r} \in \mathcal{N}}\sum_{i=1}^{N}\sum_{j=1}^{N}\sum_{s=1}^{3}r^{s}L^{s}w_{j}^{l}\frac{\partial g_{ij}^{\mathbf{r}(0)}}{\partial x^{s}},
\end{aligned}
\label{eq:u_A1_w}
\end{equation}
where we have relabeled $\mathbf{r}\leftrightarrow-\mathbf{r}$ in
the second sum and then used (\ref{eq:g_period}) in the first sum.

\begin{equation}
\begin{aligned}
\mathbf{u}^{T}B_{1}^{l}(g_{ij}^{\mathbf{r}(0)})\mathbf{u}= 
& \frac{1}{2}\sum_{\mathbf{r} \in \mathcal{N}}\sum_{j=1}^{N}\Biggl[\Biggr.(x_{j}^{l}+r^{l}L^{l})\sum_{i=1}^{N}\sum_{s=1}^{3}(x_{i}^{s}+x_{j}^{s}+r^{s}L^{s})\frac{\partial g_{ij}^{\mathbf{r}(0)}}{\partial x^{s}} \\
& \qquad \qquad \qquad \qquad \qquad \quad -x_{j}^{l}\sum_{m=1}^{N}\sum_{s=1}^{3}(x_{j}^{s}+x_{m}^{s}+r^{s}L^{s})\frac{\partial g_{jm}^{\mathbf{r}(0)}}{\partial x^{s}}\Biggr.\Biggl]\\
 = & \frac{1}{2}\sum_{\mathbf{r} \in \mathcal{N}}\sum_{j=1}^{N}\Biggl[\Biggr.(x_{j}^{l}+r^{l}L^{l})\sum_{i=1}^{N}\sum_{s=1}^{3}(x_{i}^{s}+x_{j}^{s}+r^{s}L^{s})\frac{\partial g_{ij}^{\mathbf{r}(0)}}{\partial x^{s}} \\
& \qquad \qquad \qquad \qquad \qquad \quad -x_{j}^{l}\sum_{m=1}^{N}\sum_{s=1}^{3}(x_{j}^{s}+x_{m}^{s}-r^{s}L^{s})\frac{\partial g_{jm}^{\mathbf{-r}(0)}}{\partial x^{s}}\Biggr.\Biggl]\\
& =\frac{1}{2}\sum_{\mathbf{r} \in \mathcal{N}}\sum_{j=1}^{N}\Biggl[\Biggr.(x_{j}^{l}+r^{l}L^{l})\sum_{i=1}^{N}\sum_{s=1}^{3}(x_{i}^{s}+x_{j}^{s}+r^{s}L^{s})\frac{\partial g_{ij}^{\mathbf{r}(0)}}{\partial x^{s}}\\
& \qquad \qquad \qquad \qquad \qquad \quad -x_{j}^{l}\sum_{m=1}^{N}\sum_{s=1}^{3}(x_{j}^{s}+x_{m}^{s}-r^{s}L^{s})\frac{\partial g_{mj}^{\mathbf{r}(0)}}{\partial x^{s}}\Biggr.\Biggl]\\
& =\frac{1}{2}\sum_{\mathbf{r} \in \mathcal{N}}\sum_{j=1}^{N}\Biggl[\Biggr.2x_{j}^{l}\sum_{i=1}^{N}\sum_{s=1}^{3}(r^{s}L^{s})\frac{\partial g_{ij}^{\mathbf{r}(0)}}{\partial x^{s}} \\
& \qquad \qquad \qquad \qquad \qquad \quad +r^{l}L^{l}\sum_{i=1}^{N}\sum_{s=1}^{3}(x_{i}^{s}+x_{j}^{s}+r^{s}L^{s})\frac{\partial g_{ij}^{\mathbf{r}(0)}}{\partial x^{s}}\Biggr.\Biggl]\\
& =\frac{1}{2}\sum_{\mathbf{r} \in \mathcal{N}}\sum_{i=1}^{N}\sum_{j=1}^{N}\sum_{s=1}^{3}\left[r^{l}L^{l}r^{s}L^{s}+2x_{j}^{l}r^{s}L^{s}+r^{l}L^{l}(x_{i}^{s}+x_{j}^{s})\right]\frac{\partial g_{ij}^{\mathbf{r}(0)}}{\partial x^{s}}\\
& =\frac{1}{2}\sum_{\mathbf{r} \in \mathcal{N}}\sum_{i=1}^{N}\sum_{j=1}^{N}\sum_{s=1}^{3}\left[r^{l}L^{l}r^{s}L^{s}+2x_{j}^{l}r^{s}L^{s}\right]\frac{\partial g_{ij}^{\mathbf{r}(0)}}{\partial x^{s}},
\end{aligned}
\label{eq:u_B1_u}
\end{equation}
where we have relabeled $\mathbf{r}\leftrightarrow-\mathbf{r}$ and then used (\ref{eq:g_period}) in the second sum. In the last line we made use of
\[
\sum_{\mathbf{r} \in \mathcal{N}}\sum_{i=1}^{N}\sum_{j=1}^{N}r^{l}L^{l}x_{i}^{s}\frac{\partial g_{ij}^{\mathbf{r}(0)}}{\partial x^{s}}=-\sum_{\mathbf{r} \in \mathcal{N}}\sum_{i=1}^{N}\sum_{j=1}^{N}r^{l}L^{l}x_{j}^{s}\frac{\partial g_{ij}^{\mathbf{r}(0)}}{\partial x^{s}}.
\]

\begin{equation}
\begin{aligned}
\mathbf{u}^{T}C_{0}^{kl}(g_{ij}^{\mathbf{r}(0)})\mathbf{u}= & \frac{1}{2}\sum_{\mathbf{r} \in \mathcal{N}}\sum_{i=1}^{N}\sum_{j=1}^{N}\left[(x_{j}^{k}+r^{k}L^{k})(x_{j}^{l}+r^{l}L^{l})g_{ij}^{\mathbf{r}(0)}-\delta_{ij}x_{i}^{k}x_{i}^{l}\sum_{m=1}^{N}g_{im}^{\mathbf{r}(0)}\right]\\
= & \frac{1}{2}\sum_{\mathbf{r} \in \mathcal{N}}\sum_{j=1}^{N}\left[(x_{j}^{k}+r^{k}L^{k})(x_{j}^{l}+r^{l}L^{l})\sum_{i=1}^{N}g_{ij}^{\mathbf{r}(0)}-x_{j}^{k}x_{j}^{l}\sum_{m=1}^{N}g_{jm}^{\mathbf{r}(0)}\right]\\
= & \frac{1}{2}\sum_{\mathbf{r} \in \mathcal{N}}\sum_{j=1}^{N}\left[(x_{j}^{k}+r^{k}L^{k})(x_{j}^{l}+r^{l}L^{l})\sum_{i=1}^{N}g_{ij}^{\mathbf{r}(0)}-x_{j}^{k}x_{j}^{l}\sum_{m=1}^{N}g_{jm}^{\mathbf{-r}(0)}\right]\\
= & \frac{1}{2}\sum_{\mathbf{r} \in \mathcal{N}}\sum_{j=1}^{N}\left[(x_{j}^{k}+r^{k}L^{k})(x_{j}^{l}+r^{l}L^{l})\sum_{i=1}^{N}g_{ij}^{\mathbf{r}(0)}-x_{j}^{k}x_{j}^{l}\sum_{m=1}^{N}g_{mj}^{\mathbf{r}(0)}\right]\\
= & \frac{1}{2}\sum_{\mathbf{r} \in \mathcal{N}}\sum_{i=1}^{N}\sum_{j=1}^{N}\left[r^{k}L^{k}r^{l}L^{l}+r^{l}L^{l}x_{j}^{k}+r^{k}L^{k}x_{j}^{l}\right]g_{ij}^{\mathbf{r}(0)},
\end{aligned}
\label{eq:u_C0_u}
\end{equation}
where we have went through the same steps as in (\ref{eq:u_B1_u}).

\begin{equation}
\begin{aligned}
\mathbf{u}^{T}B_{0}^{k}(g_{ij}^{\mathbf{r}(0)})\mathbf{w}^{l}= & \sum_{\mathbf{r} \in \mathcal{N}}\sum_{i=1}^{N}\sum_{j=1}^{N}\left[(x_{j}^{k}+r^{k}L^{k})g_{ij}^{\mathbf{r}(0)}-\delta_{ij}x_{i}^{k}\sum_{m=1}^{N}g_{im}^{\mathbf{r}(0)}\right]w_{j}^{l}\\
= & \sum_{\mathbf{r} \in \mathcal{N}}\left[\sum_{i=1}^{N}\sum_{j=1}^{N}w_{j}^{l}(x_{j}^{k}+r^{k}L^{k})g_{ij}^{\mathbf{r}(0)}-\sum_{m=1}^{N}\sum_{j=1}^{N}w_{j}^{l}x_{j}^{k}g_{jm}^{\mathbf{r}(0)}\right]\\
= & \sum_{\mathbf{r} \in \mathcal{N}}\left[\sum_{i=1}^{N}\sum_{j=1}^{N}w_{j}^{l}(x_{j}^{k}+r^{k}L^{k})g_{ij}^{\mathbf{r}(0)}-\sum_{m=1}^{N}\sum_{j=1}^{N}w_{j}^{l}x_{j}^{k}g_{jm}^{\mathbf{-r}(0)}\right]\\
= & \sum_{\mathbf{r} \in \mathcal{N}}\left[\sum_{i=1}^{N}\sum_{j=1}^{N}w_{j}^{l}(x_{j}^{k}+r^{k}L^{k})g_{ij}^{\mathbf{r}(0)}-\sum_{m=1}^{N}\sum_{j=1}^{N}w_{j}^{l}x_{j}^{k}g_{mj}^{\mathbf{r}(0)}\right]\\
= & \sum_{\mathbf{r} \in \mathcal{N}}\sum_{i=1}^{N}\sum_{j=1}^{N}r^{k}L^{k}w_{j}^{l}g_{ij}^{\mathbf{r}(0)},
\end{aligned}
\label{eq:u_B0_w}
\end{equation}
where we have relabeled $\mathbf{r}\leftrightarrow-\mathbf{r}$ and
then used (\ref{eq:g_period}) in the second sum.
In a similar manner we have
\begin{equation}
\mathbf{u}^{T}B_{0}^{k}(g_{ij}^{\mathbf{r}(0)})\frac{\partial\mathbf{w}^{l}}{\partial x^{k}}=\sum_{\mathbf{r} \in \mathcal{N}}\sum_{i=1}^{N}\sum_{j=1}^{N}\sum_{k=1}^{3}r^{k}L^{k}\frac{\partial w_{j}^{l}}{\partial x^{k}}g_{ij}^{\mathbf{r}(0)}.\label{eq:u_B0_dw}
\end{equation}

\vspace{3mm}
\section*{Acknowledgments}
This work was supported by the Engineering and Physical Sciences Research Council [grant number EP/X023869/1]. 

We would like to thank Thomas Babb and Jon Chapman for fruitful discussions about the work contained in this manuscript. 

\bibliographystyle{siamplain}

\newpage
\bibliography{references}
\end{document}